\documentclass[11pt]{article}

\usepackage[margin=0.9in]{geometry}
\usepackage{amsmath}
\usepackage{appendix}
\usepackage{amsthm}
\usepackage{amssymb}
\usepackage{setspace}
\usepackage{color}
\usepackage{float}
\usepackage{dsfont}
\usepackage{enumerate}
\usepackage{hyperref}
\usepackage{graphicx}
\usepackage{sidecap}
\usepackage{subfigure,wrapfig}
\usepackage{algorithm}
\usepackage{algorithmic}
\newcommand{\commentout}[1]{}
\usepackage[utf8]{inputenc}
\usepackage{tikz}
\usetikzlibrary{plotmarks}
\usepackage{wrapfig}
\usepackage[normalem]{ulem}

\newtheorem{theorem}{Theorem}[section]
\newtheorem{proposition}[theorem]{Proposition}
\newtheorem{corollary}[theorem]{Corollary}
\newtheorem{lemma}[theorem]{Lemma}
\newtheorem{assumption}[theorem]{Assumption}
\theoremstyle{definition}
\newtheorem{definition}[theorem]{Definition}
\newtheorem{remark}[theorem]{Remark}

\numberwithin{equation}{section}

\renewcommand{\subset}{\subseteq}
\renewcommand{\hat}{\widehat}
\renewcommand{\tilde}{\widetilde}
\renewcommand{\epsilon}{\varepsilon}

\def\supp{\text{supp}}
\def\OmegaRF{\text{SRF}}
\def\SRF{\text{SRF}}

\newcommand{\calA}{\mathcal{A}}

\newcommand{\calY}{\mathcal{Y}}
\newcommand{\calE}{\mathcal{E}}
\newcommand{\calF}{\mathcal{F}}

\newcommand{\calH}{\mathcal{H}}

\newcommand{\calR}{\mathcal{R}}
\newcommand{\calJ}{\mathcal{J}}
\newcommand{\calT}{\mathcal{T}}

\newcommand{\diag}{\text{diag}}

\newcommand{\beq}{\begin{equation}}
\newcommand{\eeq}{\end{equation}}

\def\({\Big(}
\def\){\Big)}
\def\C{\mathbb{C}}

\def\P{\mathcal{P}}
\def\R{\mathbb{R}}
\def\T{\mathbb{T}}
\def\Z{\mathbb{Z}}

\title{Stable super-resolution limit and smallest singular value of restricted Fourier matrices}

\date{}

\author{Weilin Li\thanks{Courant Institute of Mathematical Sciences, New York University. Email: weilinli@cims.nyu.edu}
\and Wenjing Liao\thanks{School of Mathematics, Georgia Institute of Technology. Email: wliao60@gatech.edu}}

\begin{document}

\maketitle

\begin{abstract}

We consider the inverse problem of recovering the locations and amplitudes of a collection of point sources represented as a discrete measure, given $M+1$ of its noisy low-frequency Fourier coefficients. Super-resolution refers to a stable recovery when the distance $\Delta$ between the two closest point sources is less than $1/M$. We introduce a clumps model where the point sources are closely spaced within several clumps. Under this assumption, we derive a non-asymptotic lower bound for the minimum singular value of a Vandermonde matrix whose nodes are determined by the point sources. Our estimate is given as a weighted $\ell^2$ sum, where each term only depends on the configuration of each individual clump. The main novelty is that our lower bound obtains an exact dependence on the {\it Super-Resolution Factor} $SRF=(M\Delta)^{-1}$. As noise level increases, the {\it sensitivity of the noise-space correlation function in the MUSIC algorithm} degrades according to a power law in $SRF$ where the exponent depends on the cardinality of the largest clump. Numerical experiments validate our theoretical bounds for the minimum singular value and the sensitivity of MUSIC. We also provide lower and upper bounds for a min-max error of super-resolution for the grid model, which in turn is closely related to the minimum singular value of Vandermonde matrices.

\end{abstract}

{\bf Keywords:} Super-resolution, Vandermonde matrix, Fourier matrix, minimum singular value, subspace methods, MUSIC, min-max error, sparse recovery, polynomial interpolation, uncertainty principles

{\bf 2010 Math subject classification:} 42A10, 42A15, 94A08, 94A15, 94A20


\section{Introduction}
\label{sec:introduction}

\subsection{Background}

This paper studies the inverse problem of recovering a collection of point sources from its noisy low-frequency Fourier coefficients. Suppose $S$ point sources with amplitudes $x=\{x_j\}_{j=1}^S \in \mathbb{C}^S$ are located on an unknown discrete set $\Omega=\{\omega_j\}_{j=1}^S$ in the periodic interval $\mathbb{T} = [0,1)$. This collection of points sources can be represented by a discrete measure,
\begin{equation}
\label{eq:model2}
\mu(\omega)
:=\sum_{j=1}^S x_j\delta_{\omega_j}(\omega),
\end{equation}
where $\delta_{\omega_j}$ denotes the Dirac measure supported at $\omega_j$ and $\Omega$ is the support of $\mu$, denoted $\supp(\mu)$. A uniform array of $M+1$ sensors collects measurements of the point sources. Suppose the $k$-th sensor collects the $k$-th noisy Fourier coefficient of $\mu$:
\begin{equation}
\label{eq:model1}
y_k 
:=\hat\mu(k)+\eta_k
:=\int_{\T} e^{-2\pi ik\omega} \ d\mu(\omega) + \eta_k
=\sum_{j=1}^S x_j e^{-2\pi i  k\omega_j} + \eta_k, \ \quad k=0,1,\ldots,M,
\end{equation}
where $\hat\mu$ is the Fourier transform of $\mu$ and $\eta_k$ represents some unknown noise at the $k$-th sensor. Our goal is to accurately recover $\mu$, which consists of the support $\Omega$ and the amplitudes $x\in\C^S$, from the noisy low-frequency Fourier data $y = \{y_k\}_{k=0}^M\in\C^{M+1}$. 

The measurement vector $y$ can be expressed as 
\begin{equation}
	\label{eq:model3}
	y  = \Phi x+\eta,
\end{equation}
where $\Phi$ is the $(M+1)\times S$ {\it Fourier} or {\it Vandermonde} matrix (with nodes $e^{-2\pi i\omega_j}$ on the unit circle): 
\begin{equation}
\Phi
:=\Phi(\Omega,M)
:=\Phi_M(\Omega)
:=\begin{bmatrix}
1 & 1 & \ldots & 1 \\
e^{-2\pi i\omega_1} & e^{-2\pi i\omega_2} & \ldots & e^{-2\pi i\omega_S}\\
\vdots & \vdots & \vdots & \vdots \\
e^{-2\pi iM \omega_1} & e^{-2\pi iM\omega_2} & \dots & e^{-2\pi iM\omega_S}
\end{bmatrix}. 
\label{eqphi}
\end{equation}
While it is convenient to re-formulate the measurement vector $y$ in the linear system \eqref{eq:model3}, we do not have access to the sensing matrix $\Phi$ because it depends on the unknown $\Omega$.
This inverse problem is referred to as single-snapshot spectral estimation, as only one snapshot of measurements is taken by the sensors. This problem appears in many interesting imaging and signal processing applications, including remote sensing \cite{fannjiang2010remote}, inverse scattering \cite{fannjiang2015compressive,fannjiang2010compressive}, Direction-Of-Arrival (DOA) estimation \cite{krim1996array,schmidt1986multiple} and spectral analysis \cite{stoica1997introduction}.

A key step is to estimate the support $\Omega$ and then the amplitudes $x$ can be recovered as the least-squares solution of \eqref{eq:model3}. 
In the noisy case, the stability of this inverse problem crucially depends on $\Omega$. The {\it minimum separation} of $\Omega$ has been widely used to describe the stability of this inverse problem. It is defined as
\[
\Delta
:=\Delta(\Omega)
:=\min_{1\leq j<k\leq S} |\omega_j-\omega_k|_\T,
\quad
|\omega|_\T:=\min_{n\in\Z} |\omega-n|.
\]
The Heisenberg uncertainty principle implies that the spatial and frequency localization of signals are inversely proportional. When we have access to only $M$ Fourier coefficients of $\mu$, the recovery is sensitive to noise whenever $\Delta<1/M$. In the imaging community, $1/M$ is called the {\it Rayleigh Length} ($RL$), and it is regarded as the resolution that a standard imaging system can resolve \cite{den1997resolution}. Super-resolution refers to the capability of recovering point sources when $\Delta < 1/M$.
The {\it super-resolution factor} (SRF) is $1/(M\Delta)$, which characterizes the maximum number of points in $\Omega$ that is contained in $1/M$.

The first recovery method was invented by Prony \cite{prony1795essai}. The classical Prony's method \cite{prony1795essai} can recover $\mu$ exactly in the noiseless setting, but it is very sensitive to noise \cite{batenkov2013accuracy}. Numerous modifications were attempted to improve its numerical behavior \cite{golub2003separable,beylkin2005approximation,potts2010parameter}. In the signal processing community, a class of subspace methods has been widely used in applications, including MUltiple SIgnal Classification (MUSIC) \cite{schmidt1986multiple}, Estimation of Signal Parameters via Rotational Invariance Technique (ESPRIT) \cite{kailath1989esprit}, and the Matrix Pencil Method (MPM) \cite{hua1990matrixpencil}. In the past ten years, super-resolution was addressed with an optimization approach, such as the total variation minimization (TV-min) \cite{candes2013super, fernandez2013support, azais2015spike, duval2015exact, li2017elementary}.

Existing mathematical theories on the recovery of $\mu$ can be divided into three main categories: (1) 
	The {\it well-separated case} is when $\Delta\geq 1/M$, and in which case, we say that $\Omega$ is {\it well-separated}. There are many polynomial-time algorithms that {\it provably} recover $\mu$ with high accuracy. These methods include total variation minimization (TV-min) \cite{candes2013super, fernandez2013support, azais2015spike, duval2015exact, li2017elementary}, greedy algorithms \cite{duarte2013spectral,fannjiang2012coherence,bredies2013inverse,boyd2017alternating,denoyelle2019sliding}, and subspace methods \cite{hua1990matrixpencil,fannjiang2011music,liao2016music,liao2015multi,moitra2015matrixpencil,kailath1989esprit,schmidt1986multiple}. These works address the discretization error and basis mismatch issues \cite{fannjiang2011spie, fannjiang2012coherence, chi2011basismismatch} arising in compressed sensing \cite{candes2006robust,donoho2006compressed}. 
(2) The {\it super-resolution regime} is when $\Delta<1/M$. There are two main approaches to achieve super-resolution. 
	The optimization-based methods  require certain assumptions, such as  positivity \cite{morgenshtern2016super,denoyelle2017support,morgenshtern2020super}or certain sign constraints \cite{benedetto2020super} of $\mu$.
	The classical subspace methods \cite{schmidt1986multiple,kailath1989esprit,hua1990matrixpencil} can recover complex measures and have super-resolution phenomenon. Meanwhile, there are many open problems regarding the stability of subspace methods with a single snapshot of measurements. 	
	(3) Super resolution is addressed from an information theoretic point of view in \cite{donoho1992superresolution,demanet2015recoverability},  where the point sources are located on a grid of $\mathbb{R}$ with spacing $1/N$ and measurements are the noisy continuous Fourier transform of the points. These works derived asymptotic lower and upper bounds for a min-max error, when the grid spacing is sufficiently small. Recently, \cite{batenkov2019super} studied the off-the-grid min-max error when the support contains a single cluster of nodes and other nodes are well separated.

This paper addresses two important questions: {\it (1) What is the fundamental limit of super-resolution?} To quantify the fundamental difficulty of super-resolution, we use the concept of min-max error introduced by Donoho \cite{donoho1992superresolution}. The min-max error is the reconstruction error incurred by the best possible algorithm in the worst case scenario. For technical reasons, we assume that the point sources are located on a grid with spacing $1/N$ when we study the min-max error. We follow the theme in \cite{donoho1992superresolution,demanet2015recoverability} to relate the min-max error with $\sigma_{\min}(\Phi(\Omega,M))$ for the worst subset $\Omega$ on the grid. {\it (2) What is the stability of subspace methods?} The focus of this paper is the MUSIC algorithm \cite{schmidt1986multiple}.  
MUSIC can recover {\it complex} measures and is well known for its super-resolution capability \cite{odendaal1994two}. However, there are many unanswered questions related to its stability when $\Delta < 1/M$. 
In the classical Direction-Of-Arrival setting \cite{krim1996array} with multiple snapshots of measurements, the Cram\'er-Rao bound \cite{stoica1995resolution,lee1992cramer,lee1992eigenvalues} gives a stability of MUSIC with respect to noise, $\#{\rm Snapshot}$ and the source separation $\Delta$, in the asymptotic scenario when $\#{\rm Snapshot}\rightarrow \infty$ and $\Delta\rightarrow 0$. These works can not be directly applied to the single-snapshot setting. This paper aims to establish a sensitivity analysis for the MUSIC algorithm, which is different from the Cram\'er-Rao bound since the noise statistics are not emphasized. Recently, it was shown in \cite{liao2016music} that the perturbation of the noise-space correlation function in the single-snapshot MUSIC is closely related with $\sigma_{\min}(\Phi)$. The key question is to accurately estimate $\sigma_{\min}(\Phi)$ for a given support $\Omega$.

\subsection{Contributions}

This paper has three main contributions: accurate lower bounds for $\sigma_{\min}(\Phi)$ under geometric assumptions of the support set, improvements to the min-max error of super-resolution by sparsity constraints, and a sensitivity analysis of the noise-space correlation function in the MUSIC algorithm. We informally summarize our main results here and postpone the formal definitions until later.

\begin{enumerate}[(1)]
	\item 
	We consider a geometric clumps model for $\Omega$. Assume that: (1) $\Omega$ can be written as the disjoint union of $A$ finite sets called clumps, where each clump is contained in an interval of length $1/M$ and the distance between any two clumps is at least $\beta/M$; (2) the minimum separation of $\Omega$ is at least $\alpha/M$.  
	
	Theorem \ref{thm:clump2} shows that if $\lambda_a$ denotes the cardinality of the $a$-th clump, for any $\alpha>0$ and sufficiently large $\beta>0$, there exist constants $\{C_a(\lambda_a,M)\}_{a=1}^A$ such that
	\begin{equation}
	\label{eqinformal1}
	\sigma_{\min}(\Phi)
	\geq \sqrt{M}\ \(\sum_{a=1}^A \big (C_a(\lambda_a,M) \, \alpha^{-\lambda_a+1}\big )^2 \)^{-1/2}.
	\end{equation}
	The explicit formula for $C_a$ is given in equation \eqref{thm2ca} and we discuss its dependence on $\lambda_a$ and $M$ in Remark \ref{rem:Ca}. The main novelty of this bound is that the exponent on $\alpha^{-\lambda_a+1}=\SRF^{\lambda_a-1}$ only depends on $\lambda_a$ as opposed to $S$, which shows that $\sigma_{\min}(\Phi)$ depends on the local geometry of $\Omega$. We provide an upper bound for $\sigma_{\min}(\Phi)$ in Proposition \ref{prop:upper} and numerical experiments in Section \ref{subsecboundnum} to show that the dependence on $\alpha$ is tight. 
	
	\item
	We derive an estimate of the min-max error under sparsity constraints, 
	which illustrates why geometric assumptions on $\Omega$ are both natural and necessary. Suppose $S$ and $\Delta$ are known and for technical reasons, we assume that $\Omega$ lies on a fine grid with spacing $1/N$. The min-max error $\calE(M,N,S,\delta)$ for this model is defined to be the error incurred by the best possible recovery algorithm(s) when the error is measured with respect to the worst case measure and noise $\eta$ satisfying $\|\eta\|_2\leq\delta$. 
	
	Theorem \ref{thm:minmax} provides explicit constants $A(M,S)$ and $B(M,S)$ such that 
	\begin{equation}
	\label{eqinformal2}
	A(M,S)\ \(\frac{N}{M}\)^{2S-1}\ \delta
	\leq \calE(M,N,S,\delta)
	\leq B(M,S)\ \(\frac{N}{M}\)^{2S-1}\ \delta.
	\end{equation}
	The dominant factor in both inequalities is $(N/M)^{2S-1}=\SRF^{2S-1}$. Hence, without any prior geometric assumptions on $\Omega$, no algorithm can accurately estimate every measure and noise, unless $\delta$ is smaller than $\SRF^{-2S+1}$. 
	
	\item
	We provide a sensitivity analysis of the noise-space correlation function in MUSIC under the clumps model. MUSIC amounts to recovering the point sources from the $S$ smallest local minima of a noise-space correlation function $\calR$. Theorem \ref{thmmusic} shows that when $\Omega$ satisfies the clumps assumption in Theorem \ref{thm:clump2}, then for Gaussian noise $\eta \sim \mathcal{N}(0,\sigma^2 I)$, in order to guarantee an $\epsilon$-perturbation of $\calR$ in the supremum norm, the noise-to-signal ratio that MUSIC can tolerate obeys the following scaling law,

	\begin{equation}
	\label{eqinformal3}
	\frac{\|\eta\|_2}{x_{\min}} \propto \sqrt{M} \left(
	\sum_{a=1}^A \, \big(C_a(\lambda_a,M/2) \, \alpha^{-\lambda_a+1}\big)^2
	\right)^{-1}\epsilon. 
	\end{equation}
	Our result shows that the sensitivity of the noise-space correlation function in MUSIC is exponential in $1/{\rm SRF}$, and importantly, the exponent depends on $\lambda_a$ instead of the sparsity $S$. This estimate is verified by numerical experiments. A perturbation bound for deterministic noise is given in Corollary \ref{comusicbounded}.
	 
\end{enumerate}

Let us briefly discuss the implications of our main results to super-resolution. Inequalities \eqref{eqinformal2} show that, super-resolution solely based on sparsity and minimum separation is impossible unless the noise is smaller than $\SRF^{-2S+1}$. The results can be greatly improved if the structure of $\Omega$ is exploited. Inequality \eqref{eqinformal3}, derived from \eqref{eqinformal1}, indicates that the perturbation of $\calR$ is small if the noise is smaller than $\SRF^{-2\lambda+1}$ as opposed to $\SRF^{-2S+1}$ where $\lambda=\max_a \lambda_a$ and $S=\sum_a \lambda_a$. This rigorously confirms prior numerical evidence that MUSIC can succeed in the super-resolution regime, if $\lambda$ is sufficiently small. Although our analysis only pertains to perturbations of $\calR$, this is an inherent feature of MUSIC.

\subsection{Outline}

The remainder of this paper is organized as follows. Since this paper encapsulates three main topics, we present each in its own section and related work is located in the last subsection. Our estimates for $\sigma_{\min}(\Phi)$ and proof strategy are contained in Section \ref{sec:sr} and the min-max error under sparsity constraints is studied in Section \ref{sec:minmax}. Numerical experiments are included highlighting their accuracy. Section \ref{sec:music} explains the MUSIC algorithm, and includes a new sensitivity analysis of the noise-space correlation function under the clumps model. Numerical simulations are provided to validate our theoretical analysis on MUSIC. Appendices \ref{sec:proofs}, \ref{sec:props}, and \ref{sec:lemmas} contains the proofs for all the theorems, propositions, and lemmas, respectively.


\section{Minimum singular value of Vandermonde matrices}
\label{sec:sr}

\subsection{Lower bounds under a clumps model} 

\begin{figure}
  \begin{center}
    \includegraphics[width=0.7\textwidth]{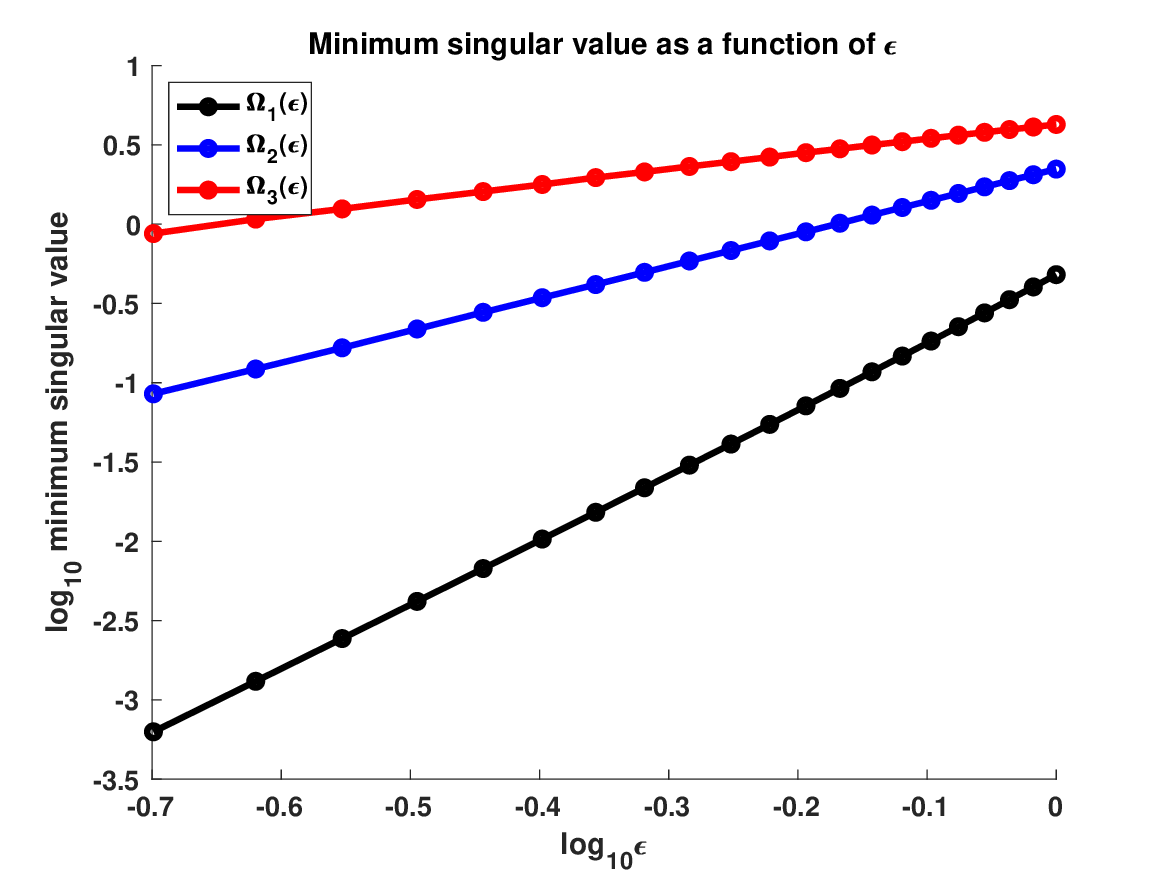}
  \end{center}
  \caption{Consider the sets $\Omega_1(\epsilon)$, $\Omega_2(\epsilon)$, $\Omega_3(\epsilon)$ defined in \eqref{eqthreesets} each with $\Delta(\epsilon)=0.01\epsilon$. The functions $\epsilon$ $\mapsto$ $\sigma_{\min}(\Phi_M(\Omega_j(\epsilon)))$ is exponential with very different exponents.} 
  \label{fig:example}
\end{figure}

The quantity $\sigma_{\min}(\Phi)$ is extremely sensitive to the ``geometry" or configuration of $\Omega$ in the super-resolution regime $\Delta<1/M$. To support this assertion, Figure \ref{fig:example} provides examples of three sets 
\begin{gather}
\label{eqthreesets}	
	\begin{split}
	\Omega_1(\epsilon) &=\epsilon\{0,0.01,0.02,0.03,0.04\}, \\ \Omega_2(\epsilon) &=\epsilon\{0,0.01,0.02,0.4,0.5\}, \\
	\Omega_3(\epsilon) &=\epsilon\{0,0.01,0.3,0.4,0.5\}
	\end{split}
\end{gather}
where $0.2\leq \epsilon\leq 1$. We fix $M=50$ and $S=5$.
These three support sets have the same cardinality and minimum separation, but the minimum singular values of their associated Vandermonde matrices have different behaviors. This simple numerical experiment demonstrates that it is impossible to accurately describe $\sigma_{\min}(\Phi)$ solely in terms of $\Delta$. A more sophisticated description of the ``geometry" of $\Omega$ is required.

We introduce a clumps model where $\Omega$ consists of well-separated subsets called {\it clumps}, where each clump contains several points that can be closely spaced.

\begin{assumption}[Clumps model]
	\label{def:clumps}
	We say that a finite set $\Omega\subset\T$ consists of $A$ clumps with parameters $(A,M,S,\beta)$ if the following hold.
	\begin{enumerate}[(1)]
	\item $\Omega$ has cardinality $S$ and can be written as a disjoint union of $A$ sets, 
	$
	\Omega=\bigcup_{a=1}^A \Lambda_a, 
	$
	where each {\it clump} $\Lambda_a$ is a finite set contained in an open interval of length $1/M$.  
	\item  If $A>1$, the distance between two clumps, defined as
	$$\text{dist}(\Lambda_m,\Lambda_n)
	:=\min_{\omega_j\in\Lambda_m,\, \omega_k\in \Lambda_n} |\omega_j-\omega_k|_{\T},$$
	satisfies $\min_{m\neq n} \text{dist}(\Lambda_m,\Lambda_n)\ge \beta/M$.
	\end{enumerate}
	Throughout the paper, we denote the cardinality of $\Lambda_a$ by $\lambda_a$.
\end{assumption}

\begin{definition}[Complexity]
	\label{def:complexity}
	For any finite set $\Omega=\{\omega_k\}_{k=1}^S\subset\T$, its {\it complexity} at $\omega_j\in\Omega$ is the quantity,
	\[
	\rho_j
	:=\rho_j(\Omega,M)
	:=\prod_{\omega_k\in\Omega\colon 0<|\omega_k-\omega_j|_\T<1/M} \frac{1}{\pi M|\omega_j-\omega_k|_\T}. 
	\]
\end{definition}

Our Theorem \ref{thm:clump1} below is our most general lower bound for $\sigma_{\min}(\Phi)$, which is given in terms of a weighted $\ell^2$ aggregate of the complexity of $\Omega$ at each point. The proof of Theorem \ref{thm:clump1} can be found in Appendix \ref{proof:clump1}. 

\begin{theorem}
	\label{thm:clump1}
	Fix positive integers $A,M,S$ with $M\geq 2S^2$. Suppose $\Omega$ satisfies Assumption \ref{def:clumps} with parameters $(A,M,S,\beta)$. If $A>1$, assume that
	\begin{equation}
	\label{eq:sep1}
	\beta
	\geq \max_{1\leq a\leq A} \max_{\omega_j\in\Lambda_a} {10 \lambda_a^{5/2} (S\rho_j)^{1/(2\lambda_a)}}.
	\end{equation}
	For each $1\leq a\leq A$, we define the constant 
	\[
	B_a
	:=B_a(\lambda_a,M)
	:=\frac{20\sqrt 2}{19} \(1-\frac{\pi^2}{3\lambda_a^2}\)^{-(\lambda_a-1)/2} \(\frac{M}{\lambda_a}\)^{\lambda_a-1} \Big\lfloor \frac{M}{\lambda_a}\Big\rfloor^{-(\lambda_a-1)}.
	\]
	Then the minimum singular value of $\Phi=\Phi(\Omega,M)$ defined in \eqref{eqphi} satisfies
	\begin{equation}
	\label{eq:clump1}
	\sigma_{\min}(\Phi)
	\geq \sqrt{M} \(\sum_{a=1}^A \sum_{\omega_j\in \Lambda_a} (B_a \lambda_a^{\lambda_a} \rho_j)^2 \)^{-1/2}. 
	\end{equation}
\end{theorem}

\begin{remark}
	\label{rem:Ba}
	The constant $B_a$ is insensitive to the geometry of each clump $\Lambda_a$ because it only depends on $M$ and $\lambda_a$, and it is also independent of $S$. The dependence of $B_a$ on $M$ is weak because 
	$$\(\frac{M}{\lambda_a}\) \Big\lfloor \frac{M}{\lambda_a}\Big\rfloor^{-1} = 1+o\(\frac{M}{\lambda_a}\) \quad \text{as} \quad \frac{M}{\lambda_a}\to\infty. $$ 
	It is possible to upper bound $B_a$ by a constant that does not depend on $M$. Clearly $B_a$ does not dependent on $M$ when $\lambda_a=1$, and for $\lambda_a\geq 2$, since $t/\lfloor t \rfloor \leq t/(t-1)\leq 2$ for  $t>1$, we see that
	\[
	B_a
	\leq 
	\frac{20\sqrt 2}{19} \(1-\frac{\pi^2}{3\lambda_a^2}\)^{-(\lambda_a-1)/2} 2^{\lambda_a-1}.
	\]
	We can think of $B_a$ as a small universal constant because  and the function $n\mapsto (1-\pi^2/(3n^2))^{-(n-1)/2}$ defined on the integers $n\geq 2$ approaches a horizontal asymptote of $1$ quickly as $n$ increases. In the regime where each $\lambda_a$ is of moderate size and $M/\lambda_a$ is large, $B_a$ is approximately $20\sqrt 2/19\approx 1.4886$. 
\end{remark}

\begin{remark}
	Although this is not the main point of the theorem, we can also apply it to the well-separated case. Assume that $\Delta\geq 10 \sqrt S/M$. Then each clump $\Lambda_a$ contains a single point, $B_a=20\sqrt 2/19$ for each $1\leq a\leq A$, and $\rho_j=1$ for each $\omega_j\in\Omega$. We readily check that the conditions of Theorem \ref{thm:clump1} are satisfied, 
	and thus,
	\[
	\sigma_{\min}(\Phi)
	\geq \frac{19}{20\sqrt 2} \sqrt{\frac{M}{S}}. 
	\]
	This shows that $\sigma_{\min}(\Phi)$ is on the order of $\sqrt{M/S}$ if $\Delta$ is about $\sqrt{S}$ times larger than $1/M$. This result is weaker than the one in \cite{moitra2015matrixpencil}, which was derived using extremal functions that are specialized to the well-separated case. {Note that $\sqrt M$ is approximately the largest $\sigma_{\min}(\Phi)$ can be because $\sigma_{\max}(\Phi)\leq \|\Phi\|_F = \sqrt{MS}$, where $\|\cdot\|_F$ is the Frobenius norm.}  
\end{remark}

Theorem \ref{thm:clump1} is our most general lower bound for $\sigma_{\min}(\Phi)$ without a minimum separation condition. The bound in Theorem \ref{thm:clump1} can be reduced to a more explicit estimate if we consider the minimum separation of $\Omega$. 

\begin{assumption}[Clumps model with a minimum separation]
	\label{def:clumps2}
	We say that a finite set $\Omega\subset\T$ satisfies a clumps  model with a minimum separation with parameters $(A,M,S,\beta,\alpha)$ if $\Omega$ satisfies Assumption \ref{def:clumps} with parameters $(A,M,S,\beta)$
and moreover, the minimum separation of $\Omega$ satisfies
$\Delta \ge \alpha/M$
with $\max_{1\le a \le A } (\lambda_a-1)<1/\alpha.$
	 \end{assumption}

\begin{figure}[h]
	\centering
	\begin{tikzpicture}[xscale = 1,yscale = 1]
	\draw[thick] (-6,0) -- (-0.5,0);
	\filldraw[red] (-5,0) circle (0.1cm);		
	\filldraw[red] (-4.7,0) circle (0.1cm);		
	\filldraw[red] (-4.4,0) circle (0.1cm);		
	\draw[blue,thick,<->] (-4.7,-0.2) -- (-4.4,-0.2);
	\node[blue,below] at (-4.4,-0.2) {$\alpha/M$};
	
	\draw[blue,thick,<->] (-4.4,0.2) -- (-2,0.2);
	\node[blue,above] at (-3.2,0.2) {$\beta/M$};
	
	\node[above] at (-4.7,0.2) {$\Lambda_1$};
	\filldraw[red] (-2,0) circle (0.1cm);		
	\filldraw[red] (-1.7,0) circle (0.1cm);		
	\filldraw[red] (-1.4,0) circle (0.1cm);		
	\draw[blue,thick,<->] (-1.7,-0.2) -- (-1.4,-0.2);
	\node[blue,below] at (-1.4,-0.2) {$\alpha/M$};
	\node[above] at (-1.7,0.2) {$\Lambda_2$};
	\draw[dotted,thick] (-0.5,0) -- (1,0);
	\draw[thick] (1,0) -- (6.4,0);
	\filldraw[red] (2,0) circle (0.1cm);		
	\filldraw[red] (2.3,0) circle (0.1cm);		
	\filldraw[red] (2.6,0) circle (0.1cm);		
	\draw[blue,thick,<->] (2.3,-0.2) -- (2.6,-0.2);
	\node[blue,below] at (2.6,-0.2) {$\alpha/M$};
	\node[above] at (2.3,0.2) {$\Lambda_{A-1}$};
	\filldraw[red] (5,0) circle (0.1cm);		
	\filldraw[red] (5.3,0) circle (0.1cm);		
	\filldraw[red] (5.6,0) circle (0.1cm);		
	\draw[blue,thick,<->] (5.3,-0.2) -- (5.6,-0.2);
	\node[blue,below] at (5.6,-0.2) {$\alpha/M$};
	\node[above] at (5.3,0.2) {$\Lambda_A$};
	\draw[blue,thick,<->] (2.6,0.2) -- (5,0.2);
	\node[blue,above] at (3.8,0.2) {$\beta/M$};
	\end{tikzpicture}
	\caption{An example of clumps model with a minimum separation.}
	\label{FigDemoClumps1}
\end{figure}

Theorem \ref{thm:clump2} below is derived from Theorem \ref{thm:clump1} under Assumption \ref{def:clumps2}, which is proved in Appendix \ref{proof:clump2}.

\begin{theorem}
	\label{thm:clump2}
	Fix positive integers $A,M,S$ with $M\geq S^2$. Suppose $\Omega$ satisfies Assumption \ref{def:clumps2} with parameters $(A,M,S,\beta,\alpha)$. 
	If $A>1$, assume that 
	\begin{equation}
	\label{eq:sep2}
	\beta
	\geq \max_{1\leq a\leq A}\  \frac{20S^{1/2}\lambda_a^{5/2}}{\alpha^{1/2}}.
	\end{equation}
	For each $1\leq a\leq A$, let $B_a:=B_a(\lambda_a,M)$ be the constant defined in Theorem \ref{thm:clump1} and
	\begin{equation}
	\label{thm2ca}
	C_a
	:=C_a(\lambda_a,M)
	:= B_a \lambda_a\(\frac{\lambda_a}{\pi}\)^{\lambda_a-1}
	\(\sum_{j=1}^{\lambda_a} \prod_{k=1,\ k\not=j}^{\lambda_a} \frac{1}{(j-k)^2}\)^{1/2}.
	\end{equation}
		Then the minimum singular value of $\Phi=\Phi(\Omega,M)$ defined in \eqref{eqphi} satisfies 
	\begin{equation}
	\label{thm2lowerbound}
	\sigma_{\min}(\Phi)
	\geq \sqrt{M}\ \(\sum_{a=1}^A \big( C_a \alpha^{-\lambda_a+1} \big)^2 \)^{-1/2}. 
	\end{equation}
\end{theorem}

\begin{remark}
	\label{remark:thm12}
	The main difference between Theorems \ref{thm:clump1} and \ref{thm:clump2} is that, Theorems \ref{thm:clump1} is more general and accurate, while Theorems \ref{thm:clump2} is more concrete since it bounds $\sigma_{\min}(\Phi)$ in terms of $\alpha=1/\SRF$. The conclusions in Theorem \ref{thm:clump1} and Theorem \ref{thm:clump2} are identical for the special case where each clump $\Lambda_a$ consists of $\lambda_a$ points consecutively spaced by $\alpha/M$. For all other configurations of $\Omega$, Theorem \ref{thm:clump1} provides a more accurate lower bound. The separation condition \eqref{eq:sep1} is always weaker than \eqref{eq:sep2}. 
\end{remark}

\begin{remark}
	\label{rem:Ca}
	According to Remark \ref{rem:Ba}, the constant $B_a$ can be thought of as a universal constant. For the behavior of $C_a$, we have a simple upper bound $C_a\leq B_a \lambda_a^{3/2} (\lambda_a/\pi)^{\lambda_a-1}$, which is a reasonable bound for small and moderate $\lambda_a$. A more refined argument (see Lemma \ref{lemma:comb} in Appendix \ref{proof:discrete}) shows that 
	\begin{equation}
	C_a 
	\leq  B_a\lambda_a \(\frac{\lambda_a}{\pi}\)^{\lambda_a-1} 2\pi e \sqrt{\lambda_a} \(\frac{\lambda_a}{2}-1\)^{-\lambda_a} e^{\lambda_a}
	=2\pi^2 e B_a\sqrt{\lambda_a} \(\frac{2e\lambda_a}{\pi(\lambda_a-2)}\)^{\lambda_a}.
	\label{eqca}
	\end{equation}
	For large $\lambda_a$, the right hand side scales like $C B_a \sqrt{\lambda_a} c^{\lambda_a}$ for universal constants $C,c>0$. 
\end{remark}

The main contribution of this theorem is the exponent on $\SRF=(M\Delta)^{-1}=1/\alpha$, which depends on $\lambda_a$ as opposed to the sparsity $S$. Let us look at a special case of $\Omega$, where each clump $\Lambda_a$ contains $\lambda$ points equally spaced by $\alpha/M$ and the distance between clumps is $\beta/M$ where $\beta$ is properly chosen such that \eqref{eq:sep2} holds. See Figure \ref{FigDemoClumps1} for an illustration.

In this case, Theorem \ref{thm:clump2} implies
\begin{equation}
\label{eqlowert1}
\sigma_{\min}(\Phi(\Omega,M)) \ge
C(\lambda) A^{-1/2} \underbrace{\sqrt{M}}_{\text{scaling}}  \cdot \underbrace{\alpha^{\lambda-1}}_{(1/{\rm SRF})^{\lambda-1}},
\end{equation}
where the constant $C(\lambda)$ only  depends on $\lambda$. Here, $\sqrt M$ is a natural scaling factor because each column of $\Phi(\Omega,M)$ has Euclidean norm $\sqrt{M+1}$. Importantly, the lower bound scales like $\alpha^{\lambda-1}=(1/{\rm SRF})^{\lambda -1}$	where $\lambda$ is the cardinality of each clump instead of $S$, which matches our intuition that the conditioning of $\Phi(\Omega,M)$ should only depend on how complicated each individual clump is.

\subsection{Upper bound for the minimum singular value}

Our lower bounds in Theorems \ref{thm:clump1} and \ref{thm:clump2} show that, for the special support in Figure \ref{FigDemoClumps1}, our lower bound for $\sigma_{\min}(\Phi)$ depends linearly on $\SRF^{-\lambda+1}$. Proposition \ref{prop:upper} below proves that this dependence is sharp. Its proof can be found in Appendix \ref{proof:upper} and uses a method similar to one in \cite{donoho1992superresolution}. 

\begin{proposition}
	\label{prop:upper}
	Fix positive integers $M,S,\lambda$ such that $\lambda\leq S\leq M-1$. Let $\omega\in\T$ and $\alpha>0$ such that
	\begin{equation}
	\label{eq:alpha}
	\alpha\leq \frac{1}{C(\lambda) \sqrt{M+1}},
	\quad\text{where}\quad
	C(\lambda) = 2\pi \sum_{j=0}^{\lambda-1} {\lambda-1 \choose j} \frac{j^\lambda}{\lambda!}. 
	\end{equation}
	Assume that $\Omega=\{\omega_j\}_{j=1}^S\subset\T$ contains the set,
	\[
	\Lambda
	=\omega+\Big\{0,\frac{\alpha}{M},\dots,\frac{(\lambda-1)\alpha}{M} \Big \}.
	\]
	Let $\Phi=\Phi(\Omega,M)$ be the $(M+1)\times S$ Vandermonde matrix associated with $\Omega$ and $M$. Then
	\[
	\sigma_{\min}(\Phi)
	\leq {2\lambda-2 \choose \lambda-1}^{-1/2} 2\sqrt{M+1}\ (2\pi \alpha)^{\lambda-1}.
	\]
\end{proposition}

Proposition \ref{prop:upper} shows that if $\Omega$ contains a set $\Lambda$, which consists of $\lambda$ points equally spaced by $\alpha/M$ for a sufficiently small $\alpha$, then $\sigma_{\min}(\Phi)$ depends on $\alpha^{\lambda-1}$. This implies that the dependence on $\alpha$ in the order of $\alpha^{\lambda-1}$ in our lower bound of \eqref{eqlowert1} is tight.

\subsection{Numerical accuracy of Theorems \ref{thm:clump1} and \ref{thm:clump2}}
\label{subsecboundnum}

To numerically evaluate the accuracy of Theorems \ref{thm:clump1} and \ref{thm:clump2}, we consider the case where $\Omega$ consists of $A$ clumps, each clump contains $\lambda$ equally spaced points separated by $\alpha/M$, and the clump separation is $\beta/M$. We fix $M$, vary $\SRF=1/\alpha$, and select $\beta$ to be the right hand side of inequality \eqref{eq:sep2}. As discussed in Remark \ref{remark:thm12}, both Theorem \ref{thm:clump1} and \ref{thm:clump2} provide the identical lower bound for $\sigma_{\min}(\Phi)$ for this example of $\Omega$, which is of the form
\begin{equation}
\label{eq:true}
\sigma_{\min}(\Phi)
\geq C(\lambda) \sqrt{\frac{M}{A}} \, \SRF^{-\lambda+1}
=:\phi(A,M,\lambda,\SRF).
\end{equation}

Figure \ref{fig:exper1} displays plots of $\sigma_{\min}(\Phi)$ and $\phi$ as functions of $\SRF$, for several values of $A$ and $\lambda$. Their slopes are identical, which establishes that the dependence on $\SRF$ in inequality \eqref{eq:true} is exact. In other words, the experiments verify that $\sigma_{\min}(\Phi)$ should only depend on the cardinality of each clump and not on the total number of points $S$ under the theorems' assumptions.

\begin{figure}
	\hspace{-0.5cm}
	\subfigure[2-clump model: $A=2$]{\includegraphics[width=0.55\textwidth]{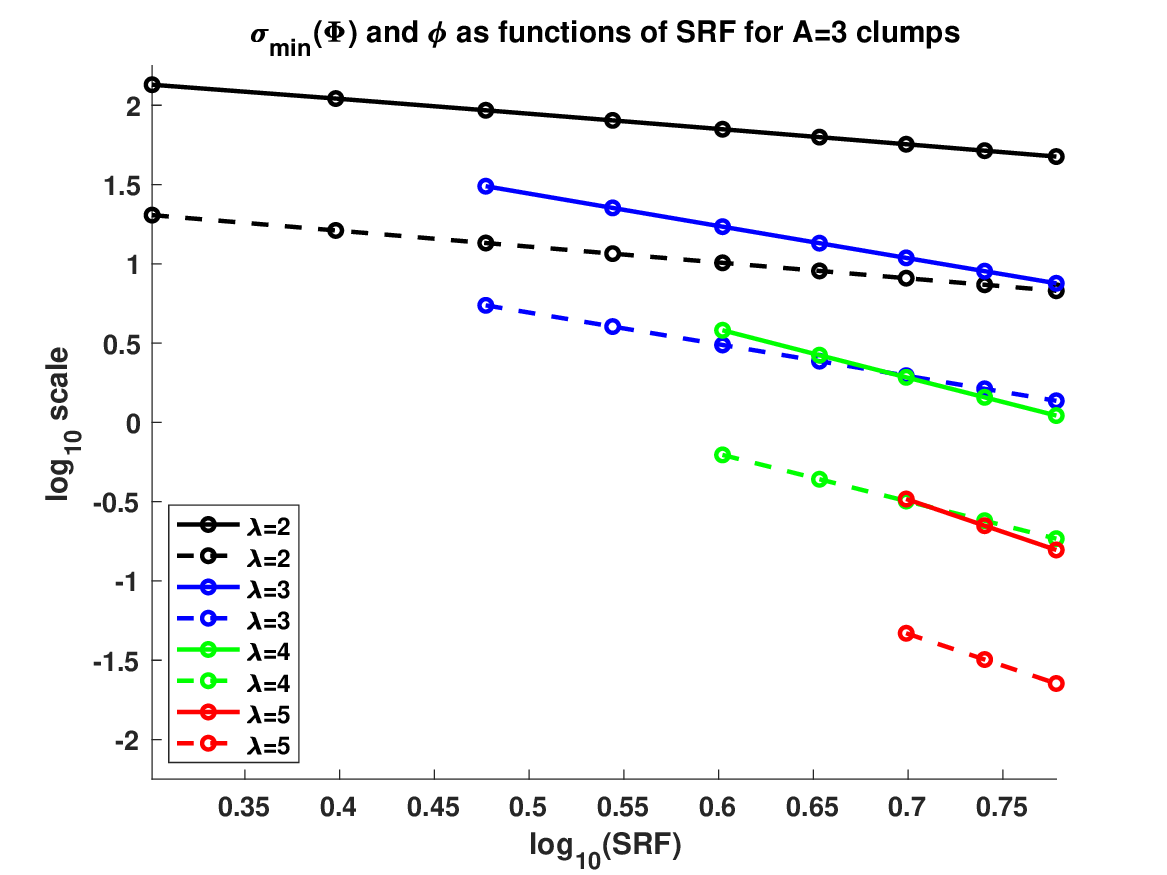}}
	\hspace{-1cm}
	\subfigure[3-clump model: $A=3$]{\includegraphics[width=0.55\textwidth]{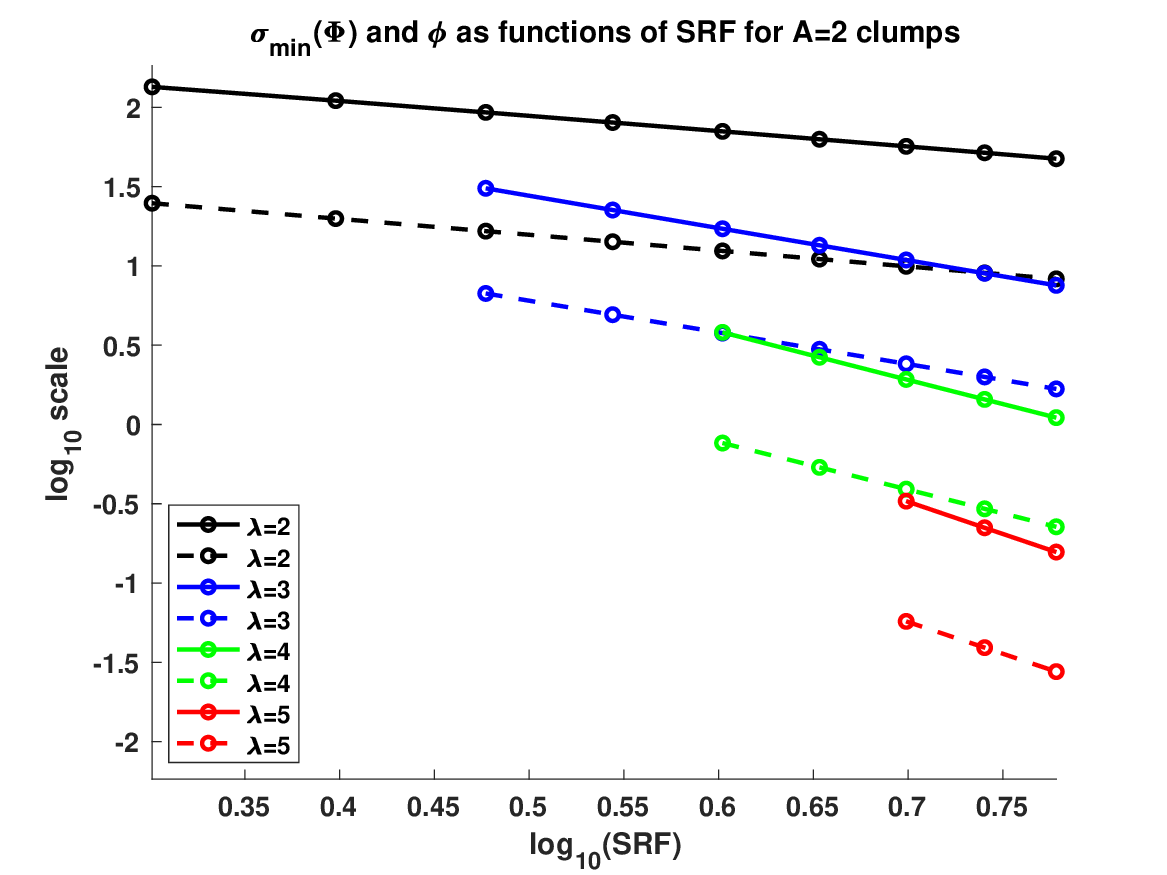}}
	\caption{Plots of $\sigma_{\min}(\Phi)$ (solid lines) and $\phi$ in \eqref{eq:true} (dashed lines), as functions of $\SRF=1/\alpha$ when the other parameters $A,\lambda,M$ are fixed. For all the curves, we set $M=50000$ and let $\Omega$ consist of $A$ clumps, where each clump contains $\lambda$ points equally separated by $\alpha/M$ and the clump separation is given by the right hand side of inequality \eqref{eq:sep2}. We consider the following range of parameters: $2\leq A\leq 3$, $2\leq \lambda \leq 5$, and $\lambda\leq \SRF\leq 6$.}
	\label{fig:exper1}  
\end{figure}

\subsection{Proof strategy by polynomial duality and interpolation}
\label{sec:duality}

Our primary method for lower bounding $\sigma_{\min}(\Phi)$ is through a dual characterization via trigonometric interpolation. We begin with some notation and definitions. Let $\P(M)$ be the space of all smooth functions $f$ on $\T$ such that for all $\omega\in\T$, 
\[
f(\omega)=\sum_{m=0}^{M} \hat f(m)e^{2\pi im\omega}.
\]
We call $f$ a {\it trigonometric polynomial} of degree at most $M$. 

\begin{definition}[Polynomial interpolation set]
	Given $\Omega=\{\omega_j\}_{j=1}^S\subset\T$ and $v\in\C^S$, the polynomial interpolation set, denoted by $\P(\Omega,M,v)$, is the set of $f\in\P(M)$ such that $f(\omega_j)=v_j$ for each $1\leq j\leq S$. 
\end{definition}

We have the following duality between the minimum singular value of Fourier matrices and the polynomial interpolation set, and its proof is in Appendix \ref{proof:duality}. 

\begin{proposition}[Exact duality]
	\label{prop:duality}
	Fix positive integers $M$ and $S$ such that $S\leq M-1$. For any set $\Omega=\{\omega_j\}_{j=1}^S\subset\T$, let $\Phi=\Phi(\Omega,M)$ be the $(M+1)\times S$ Vandermonde matrix associated with $\Omega$ and $M$. For any $w\in\C^S$, the set $\P(\Omega,M,w)$ is non-empty. If $\sigma_{\min}(\Phi)=\|\Phi v\|_2$ for some unit norm vector $v\in\C^S$, then
	\[
	\sigma_{\min}(\Phi)
	=\max_{f\in \mathcal{P}(M,\Omega,v)} \|f\|_{L^2(\T)}^{-1}.
	\]
\end{proposition} 

There are two main technical difficulties with using the proposition. First, in the extreme case that $S\ll M$, we expect $\P(\Omega,M,v)$ to contain a rich set of functions. However, we do not know much about this set, aside from it being convex. Moreover, this set is extremely dependent on $\Omega$ because we know that $\sigma_{\min}(\Phi)$ is highly sensitive to the configuration of $\Omega$. Second, we have no information about $v$, a right singular vector associated with $\sigma_{\min}(\Phi)$. Yet, in order to invoke the duality result, we must examine the set $\P(\Omega,M,v)$.

It turns out that we can circumvent both of these issues, but doing so will introduce additional technicalities and difficulties. Proposition \ref{prop:duality2} below is a relaxation of the exact duality in Proposition \ref{prop:duality}, which will provide us with an extra bit of flexibility. Its proof can be found in Appendix \ref{proof:duality2}. 

\begin{proposition}[Robust duality]
	\label{prop:duality2}
	Fix positive integers $M$ and $S$ such that $S\leq M-1$. For any set $\Omega=\{\omega_j\}_{j=1}^S\subset\T$, unit norm vector $v\in\C^S$, and $\epsilon\in\C^S$ with $\|\epsilon\|_2< 1$, if there exists $f\in\P(M)$ such that $f(\omega_j)=v_j+\epsilon_j$ for each $1\leq j\leq S$, then 
	\[
	\|\Phi v\|_2
	\geq (1-\|\epsilon\|_2) \|f\|_{L^2(\T)}^{-1}.
	\]
\end{proposition}

In order to use these results to derive a lower bound for $\sigma_{\min}(\Phi)$ for a given $\Omega$, for each unit norm $v\in\C^S$, we construct a $f_v\in\P(\Omega,M,v)$ and then bound $\|f_v\|_{L^2(\T)}$ uniformly in $v$. This process must be done carefully; otherwise we would obtain a loose lower bound for $\sigma_{\min}(\Phi)$. This construction is technical and our approach is inspired by uncertainty principles for trigonometric polynomials \cite{turan1946rational} and uniform dilation problems on the torus \cite{alon1992uniform,konyagin2000uniformly}.

\subsection{Related work on the conditioning of Vandermonde matrices}
\label{sec:relatedVander}

The spectral properties of a Vandermonde matrix greatly depend on its nodes. Real \cite{gautschi1987lower,beckermann2000condition,eisinberg2001rectangular,eisinberg2001vandermonde}, random \cite{ryan2009asymptotic,tucci2011eigenvalue,tucci2014asymptotic}, and those located within the unit complex disk \cite{bazan2000conditioning,aubel2017vandermonde} have been studied. In this paper, we study tall and deterministic Vandermonde matrices whose nodes are on the complex unit circle, which are generalizations of harmonic Fourier matrices.

In the well-separated case $\Delta\geq 1/M$, Vandermonde matrices are well-conditioned \cite{ferreira1999super,berman2007perfect,liao2016music,moitra2015matrixpencil}. There are fewer available results for the super-resolution regime $\Delta\leq 1/M$. On one extreme, there exists a set $\Omega$ for which the conditioning of $\Phi$ scales linearly in $\SRF^{S-1}$ when all other parameters are fixed \cite{moitra2015matrixpencil}. On the other hand, Theorems \ref{thm:clump1} and \ref{thm:clump2} show that this is overly pessimistic under appropriate clumps model. 

Classical work, such as \cite{gautschi1962inverses}, primarily focused on square Vandermonde matrices. One of the earliest results for rectangular ones that also incorporates geometric information follows by combining \cite[Theorem 1]{bazan2000conditioning} and \cite[Theorem 1]{gautschi1962inverses}, which yields the estimate
\begin{equation}
\label{eq:classical}
\sigma_{\min}(\Phi)
\geq \sqrt{\frac{M}{S}} \min_{1\leq j\leq S} \prod_{k=1,\ k\not=j}^S \frac{|\omega_j-\omega_k|_\T}{2}.
\end{equation}
When $\Omega$ consists of separated clumps with parameters $(A,M,S,\beta)$, this inequality yields
\[
\sigma_{\min}(\Phi)
\geq 2^{-S} \sqrt{\frac{M}{S}} \, \min_{1\leq a\leq A} \(\frac{\beta}{M}\)^{S-\lambda_a} \, \Delta^{\lambda_a-1}.
\]
When $M$ is large, this lower bound is significantly worse than the inequality in Theorem \ref{thm:clump2} which depends on terms involving $\SRF^{-\lambda_a+1}=(M\Delta)^{\lambda_a-1}$. 

The preceding discussion is closely related to an important difference between the complexity $\rho_j$ and $\prod_{j\not=k} |\omega_j-\omega_k|^{-1}$, the reciprocal of the product term in \eqref{eq:classical}. Notice that $\rho_j$ is local in the sense it only depends on the structure of $\Omega$ in a $1/M$ neighborhood of $\omega_j$ and it includes the $\pi M$ scaling term. We include the $\pi$ factor because we use the Fourier transform with the $2\pi$ convention in \eqref{eq:model1}, but the $M$ factor is very important. In the super-resolution and imaging communities, it is known that $\SRF=(M\Delta)^{-1}$, not $\Delta^{-1}$, is the correct quantity to describe the feasibility of super-resolution. The complexity can be viewed as a refined notion of $\SRF^S=(M\Delta)^{-S}$ since $\rho_j$ only depends on the structure of $\Omega$ near $\omega_j$.

In the process of revising the first version of this manuscript, \cite{batenkov2020conditioning} independently derived lower bounds for $\sigma_{\min}(\Phi)$ with clustered nodes. Here the differences between their \cite[Corollary 3.6]{batenkov2020conditioning} and our Theorem \ref{thm:clump2}: (1) They assume that $\Omega$ consists of clumps that are all contained in an interval of length approximately $1/S^2$. For ours, the clumps can be spread throughout $\T$ and not have to be concentrated on a sub-interval. (2)	They require the aspect ratio $M\ge 4S^3$, whereas we only need $M \ge 2S^2$. They also require an upper bound on $M/S$, which prohibits their Vandermonde matrix from being too tall. 
(3) Their lower bound is of the form, $\sigma_{\min}(\Phi)\geq C_S\sqrt{M} \cdot \SRF^{-\lambda+1}$, where $C_S$ depends only on $S$ and scales like $S^{-2S}$. According to \eqref{eqca}, our constant $C_a$ scales like $ \sqrt{\lambda_a}\, c^{\lambda_a}$ for a universal constant $c>0$. (4) Our clump separation condition \eqref{eq:sep2} in Theorem \ref{thm:clump2} depends on $\alpha=\SRF^{-1}$, which means that for large $\SRF$, the clumps have to be significantly separated. In contrast, their clump separation condition only depends on $S$ but all of their clumps must remain in an interval with length approximately $1/S^2$. 

During the review period of this paper, the singular values of the same  Vandermonde matrices with nearly colliding pairs were analyzed in \cite{kunis2020condition} and \cite[Corollary 4.2]{diederichs2019well}. The nearly colliding pairs form a special case of the clumps model when each clump has at most two closely spaced points.  In this special case, the separation condition between the clumps are improved to be independent of $\Delta$ in \cite{kunis2020condition,diederichs2019well}.
A recent work \cite{kunis2020smallest} refined the main strategy in the proofs of Theorems \ref{thm:clump1} and \ref{thm:clump2} to improve the constant $C_a$ and strengthen our $\ell^2$ bound to an $\ell^\infty$ one. They also provided an extension to the multi-dimensional case. Another recent paper \cite{batenkov2020spectral} provides bounds for all singular values of Vandermonde matrices under a clumps model.

Lower bounds for the conditioning of arbitrarily sized $p\times q$ cyclically contiguous sub-matrices of the $N\times N$ discrete Fourier transform matrix was derived in \cite{barnett2020exponentially}. This sub-matrix is equivalent to our $\Phi(\Omega,p)$, where $\Omega:=\{k/N\}_{k=1}^q$. A comparison between our Proposition \ref{prop:upper} and the main result in \cite{barnett2020exponentially} can be found in the referenced paper.


\section{Min-max error and worst case analysis}
\label{sec:minmax}

\subsection{A grid model and the min-max error}

In order to understand the fundamental limits of super-resolution without any geometric information about $\Omega$, in this section, we study the min-max error under a grid and sparsity assumption on $\Omega$. Our results will illustrate that it is not only natural, but also necessary  to take the geometric information of $\Omega$ into account.

Suppose $\Omega$ has cardinality $S$ and is a subset of $\{n/N\}_{n=0}^{N-1}$, which we refer to as the grid with spacing $1/N$. In the super-resolution literature, this is called the {\it on-the-grid} model. This assumption implicitly places a minimum separation requirement so that all point sources are separated by at least $1/N$ and $\SRF=N/M$. In comparison to the clumps model, the one considered in this section only has a sparsity constraint and places no geometric constraints on the support set. The grid assumption is purely for technical reasons.  

We will define a min-max error proposed in \cite{donoho1992superresolution}. Fix positive integers $M,N,S$ such that $S\leq M\leq N$ and let $\delta>0$. We define a set of vectors in $\C^{M+1}$,
\begin{align*}
\calY 
&:= \calY(M,N,S,\delta) \\
&=\Big\{y\in\C^{M+1}: \exists\, \mu \text{ supported in } \Omega \subset \{n/N\}_{n=0}^{N-1}, \ |\Omega| = S,  \(\sum_{k=0}^M |y_k-\hat\mu(k)|^2\)^{1/2} \leq \delta \Big \}. 
\end{align*}
Let $\calA=\calA(M,N,S,\delta)$ be the set of functions $\varphi$ that maps each $y\in \calY$ to a discrete measure $\varphi_y$ supported on the grid with spacing $1/N$.

\begin{definition}[Min-max error]
	The $\ell^2$ {\it min-max error} for the on-the-grid model is
	\[
	\calE(M,N,S,\delta)
	:=\inf_{\varphi\in\calA}\ \sup_{y\in \calY(M,N,S,\delta)}  \(\sum_{n=0}^{N-1} \Big|\varphi_y\( \frac{n}{N}\)-\mu\( \frac{n}{N}\) \Big|^2\)^{1/2}.
	\]
\end{definition}

We interpret $\delta$ as the noise level and $\calY$ as the ``signal space" formed by all possible $\delta$ perturbations of the Fourier coefficients of a measure with at most $S$ Dirac masses on the grid. We interpret a function $\varphi\in \calA$ as an ``algorithm" that maps a given $y\in \calY$ to a measure $\varphi_y$ that approximates $\mu$. By taking the infimum over all possible algorithms (which includes those that are computationally intractable), the min-max error is the reconstruction error incurred by the best algorithm, when measured against the worst case signal and noise. This quantity describes the fundamental limit of super-resolution under sparsity constraints, and no algorithm can perform better than the min-max rate. 

\subsection{Lower and upper bounds of the min-max error}

To lower and upper bound the min-max error, 
we follow the approach of \cite{demanet2015recoverability} to connect the min-max error with $\sigma_{\min}(\Phi(\Omega,M))$ for the worst subset $\Omega$ on the grid.

\begin{definition}[Lower restricted isometry constant]
	Fix positive integers $M,N,S$ such that $S\leq M\leq N$. The lower restricted isometry constant of order $S$ is defined as
	\[
	\Theta(M,N,S)
	:=\min_{\Omega \subset \{n/N\}_{n=0}^{N-1}, \ |\Omega| = S} \sigma_{\min}(\Phi(\Omega,M)).
	\]
\end{definition}

\begin{remark}
	While $\Theta(M,N,S)$ is related to the lower bound of the $S$-{\it restricted isometry property} (RIP) in compressive sensing \cite{candes2008restricted}, there is a major difference. If we select $M$ appropriately chosen rows of the $N\times N$ Discrete Fourier Transform (DFT), then every $M\times S$ sub-matrix is well-conditioned, see \cite{rudelson2008sparse}. However, $\Phi(\Omega,M)$ uses the first $M$ rows of the DFT matrix, so $\Phi(\Omega,M)$ may be ill-conditioned. 
\end{remark}

The following result establishes the relationship between the min-max error and the lower restricted isometry constant. An analogue of this result for a similar super-resolution problem on $\R$ was proved in \cite{demanet2015recoverability}. Their proof carries over to this discrete setting with minor modifications. To keep this paper self-contained, we prove it in Appendix \ref{proof:demanet} and it relies on the grid assumption.  

\begin{proposition}
	[Min-max error and lower restricted isometry constant]
	\label{prop:demanet}
	Fix positive integers $M,N,S$ such that $2S\leq M\leq N$, and let $\delta>0$. Then,
	\[
	\frac{\delta}{2\Theta(M,N,2S)}
	\leq \calE(M,N,S,\delta)
	\leq \frac{2\delta}{\Theta(M,N,2S)}. 
	\]
\end{proposition}

With this result at hand, our main focus is to derive an accurate lower bound for $\Theta(M,N,S)$. We suspect that the minimum in $\Theta(M,N,S)$ occurs when all the points in $\Omega$ are consecutively spaced by $1/N$. 
If we could prove this conjecture, then we could simply apply Theorem \ref{thm:clump2} for the single clump case to derive an upper bound of $\Theta(M,N,S)$. Without this conjecture, our strategy is to apply our duality techniques in Section \ref{sec:duality} for all ${N\choose S}$ possible choices of $\Omega$ and then uniformly bound over all these possibilities. There are exponentially many possible such $\Omega$, so we must be extremely careful with the estimate in order to obtain an accurate lower bound for $\Theta(M,N,S)$. The following theorem is proved in Appendix \ref{proof:theta}.

\begin{theorem}
	\label{thm:theta}
	Fix positive integers $M,N,S$ such that $S\geq 2$, $M\geq 2S$, and $N\geq \pi MS$. We define the constant,
	\[
	C(M,S)
	:=\(\frac{12-\pi^2}{24}\)^{1/2} \(\sum_{j=1}^S  \prod_{k\not=j, \, k=1}^S \frac{1}{(j-k)^2}\)^{-1/2} \frac{1}{\sqrt S} \(\frac{\pi}{S}\)^{S-1} \( \frac{M}{S} \)^{-(S-1)} \Big\lfloor \frac{M}{S} \Big\rfloor^{S-1}. 
	\]
	Then we have
	\begin{equation}
	\label{thm3lower}
	\Theta(M,N,S)
	\geq C(M,S) \sqrt{M} \(\frac{M}{N}\)^{S-1}.
	\end{equation}
\end{theorem}

\begin{remark}
	As described in Remark \ref{rem:Ba}, when $M$ is significantly larger than $S$, the $(M/S)^{-(S-1)} \lfloor M/S \rfloor^{S-1}$ factor in $C(M,S)$ can be safely ignored if $M/S$ is sufficiently large. Alternatively, we can lower bound $C(M,S)$ independent of $M$ since $t/\lfloor t\rfloor \leq t/(t-1)\leq 2$ for all $t>2$ and so
	\[
	C(M,S)
	\geq \(\frac{12-\pi^2}{24}\)^{1/2} \(\sum_{j=1}^S  \prod_{k\not=j, \, k=1}^S \frac{1}{(j-k)^2}\)^{-1/2} \frac{1}{\sqrt S} \(\frac{\pi}{S}\)^{S-1} 2^{-(S-1)}.
	\]
	Now we concentrate on the dependence of $C(M,S)$ on $S$. Since $\prod_{k\not=j, \, k=1}^S (j-k)^{-2}\leq 1$ for each $1\leq j\leq S$, we have the lower bound,
	\[
	C(M,S)
	\geq \frac{1}{\pi} \(\frac{12-\pi^2}{24}\)^{1/2} \(\frac{\pi}{S}\)^S \( \frac{M}{S} \)^{-(S-1)} \Big\lfloor \frac{M}{S} \Big\rfloor^{S-1}.
	\]
	For large $S$, a more accurate bound is (see Lemma \ref{lemma:comb} in Appendix \ref{proof:discrete})
	\[
	C(M,S)
	\geq \frac{1}{2\pi^2 e } \(\frac{12-\pi^2}{24}\)^{1/2} e^{-S} \( \frac{S}{2}-1\)^S\(\frac{\pi}{S}\)^{S} \( \frac{M}{S} \)^{-(S-1)} \Big\lfloor \frac{M}{S} \Big\rfloor^{S-1}.
	\]
	For large $S$, the right hand side scales as $Cc^{-S}$ for universal constants $C,c>0$. 
\end{remark}

\begin{remark}
	Although we do not identify which $\Omega$ achieves the minimum in $\Theta(M,N,S)$, the theorem suggests that the minimum is attained by translates of $\Omega_*:=\{n/N\}_{n=0}^{S-1}$. Indeed, the assumption $N\geq \pi MS$ implies that $\Omega_*$ is contained in an open interval of length $1/M$. We readily check that if $\rho_j=\rho_j(\Omega_*,M)$ is the complexity of $\omega_j\in\Omega_*$, then
	\[
	\(\sum_{j=1}^S \rho_j^2\)^{1/2}
	=\(\sum_{j=1}^S \prod_{k\not=j} \frac{1}{(j-k)^2}\)^{1/2} \(\frac{N}{\pi M}\)^{S-1}.
	\]
	Notice that this is precisely the quantities that appear in Theorem \ref{thm:theta}. 
	\label{remarkminmaxworst}
\end{remark}
We next present the lower and upper bounds for the min-max error, which is proved in Appendix \ref{proof:minmax}. 

\begin{theorem}
	\label{thm:minmax}
	Fix positive integers $S,M,N$ and let $\delta>0$. 
	\begin{enumerate}[(a)]
		\item 
		Assume that $M\geq 4S$ and $N\geq 2\pi MS$, and let $C(M,S)$ be the constant defined in Theorem \ref{thm:theta}. Then, we have the upper bound,
		\[
		\calE(M,N,S,\delta)
		\leq \frac{2\delta}{C(M,2S)} \frac{1}{\sqrt M} \(\frac{N}{M}\)^{2S-1}. 
		\]
		\item 
		Assume that $M\geq 2S+1$, and $N/M\geq 2\pi C(2S)\sqrt{M+1}$, where $C(2S)$ is the constant defined in Proposition \ref{prop:upper}. Then, we have the lower bound,
		\[
		\calE(M,N,S,\delta)
		\geq \frac{\delta}{4}{4S-2\choose 2S-1}^{1/2} \frac{1}{\sqrt{M+1}} \(\frac{N}{2\pi M}\)^{2S-1}.
		\]
	\end{enumerate}
\end{theorem}

\subsection{Numerical accuracy of Theorem \ref{thm:theta}}

Let $\theta(M,N,S)$ denote the right hand side of \eqref{thm3lower}. Note that $\theta$ is only defined for $N/M\geq \pi S$. We make two important observations before numerically evaluating the accuracy of Theorem \ref{thm:theta}. 

First, while we would like to compare $\theta$ directly with $\Theta$, this is not computationally feasible because we would need to enumerate through all possible $\Omega$, for numerous values of $M,N,S$. Instead, we compare $\theta$ with the quantity,
\[
\psi(M,N,S):=\sigma_{\min}(\Phi(\Omega_*,M)),
\]
where $\Omega_*:=\{n/N\}_{n=0}^{S-1}$. Note that $\psi$ serves as a useful substitute for $\Theta$ because of the inequalities
\begin{equation}
\label{eq:ratio}
\theta(M,N,S)
\leq \Theta(M,N,S)
\leq \psi(M,N,S).
\end{equation}	
	
Second, while both $\theta$ and $\psi$ depend on three parameters, Theorem \ref{thm:theta} and Proposition \ref{prop:upper} suggest that after safely ignoring the $\sqrt M$ scaling factor, they should only depend on two parameters, the super-resolution factor $\SRF=N/M$ and the sparsity $S$. Additionally, we can only reliably perform the experiments for modest size of $\SRF^{S-1}$, or else numerical round off errors become significant.

Figure \ref{fig:exper2} (a) displays the values of $\theta$ and $\psi$ as functions of $\SRF$. Since the slope of both curves are identical, they verify that Theorem \ref{thm:theta} correctly quantifies the dependence of the min-max error on $\SRF$. Figure \ref{fig:exper2} (b) displays the ratio between $\psi$ and $\theta$. As a consequence of inequalities \eqref{eq:ratio}, this experiment also indirectly provides us information about the ratio of $\Theta$ and $\theta$. The lines in Figure \ref{fig:exper2} (b) are horizontal, which confirms our theory that, if $\SRF\geq \pi S$, then there exists a $c(M,S)>0$ such that
\[
\Theta(M,N,S)
=c(M,S) \sqrt{M} \(\frac{M}{N}\)^{S-1}.
\]

\begin{figure}[t]
	\hspace{-0.5cm}
	\subfigure[$\theta(M,N,S)$ and $\psi(M,N,S)$]{
		\includegraphics[width=0.6\textwidth]{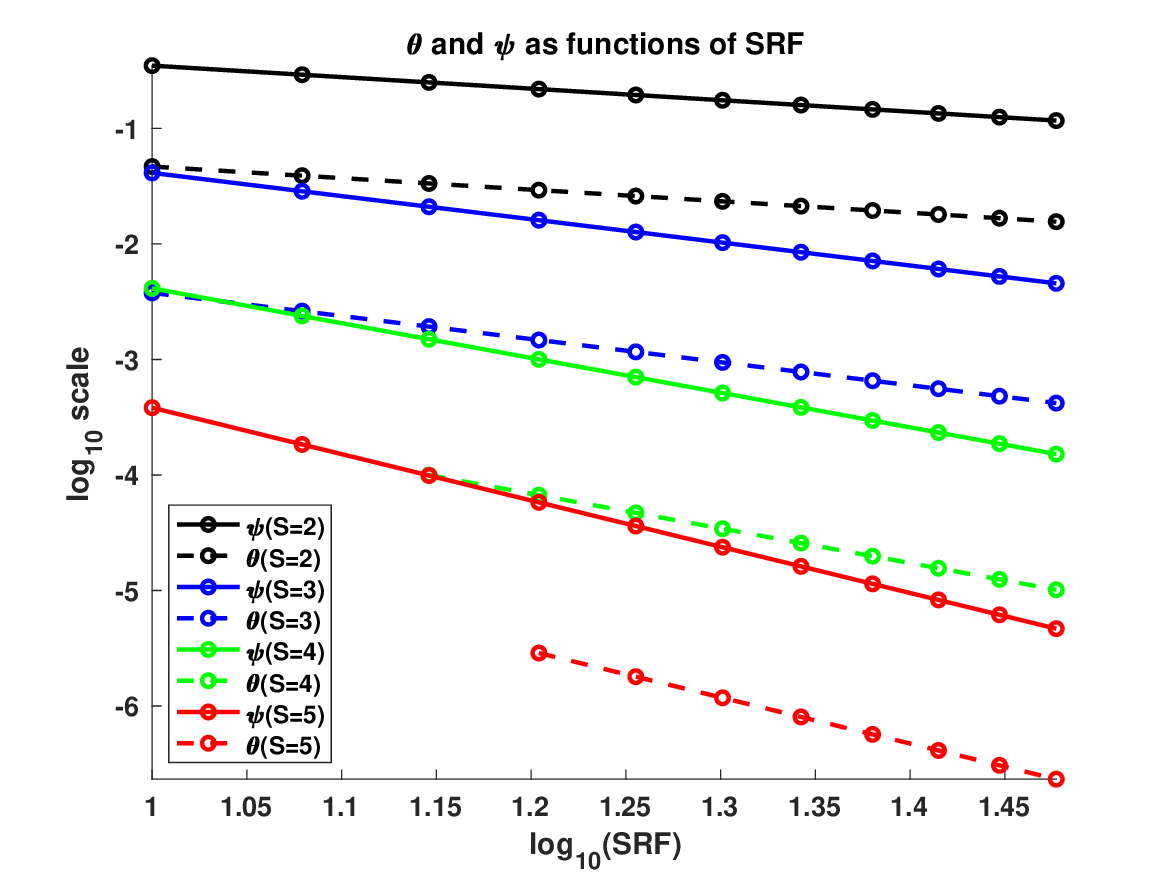}}
	\hspace{-3cm}
	\subfigure[$\psi(M,N,S)/\theta(M,N,S)$]{
		\includegraphics[width=0.6\textwidth]{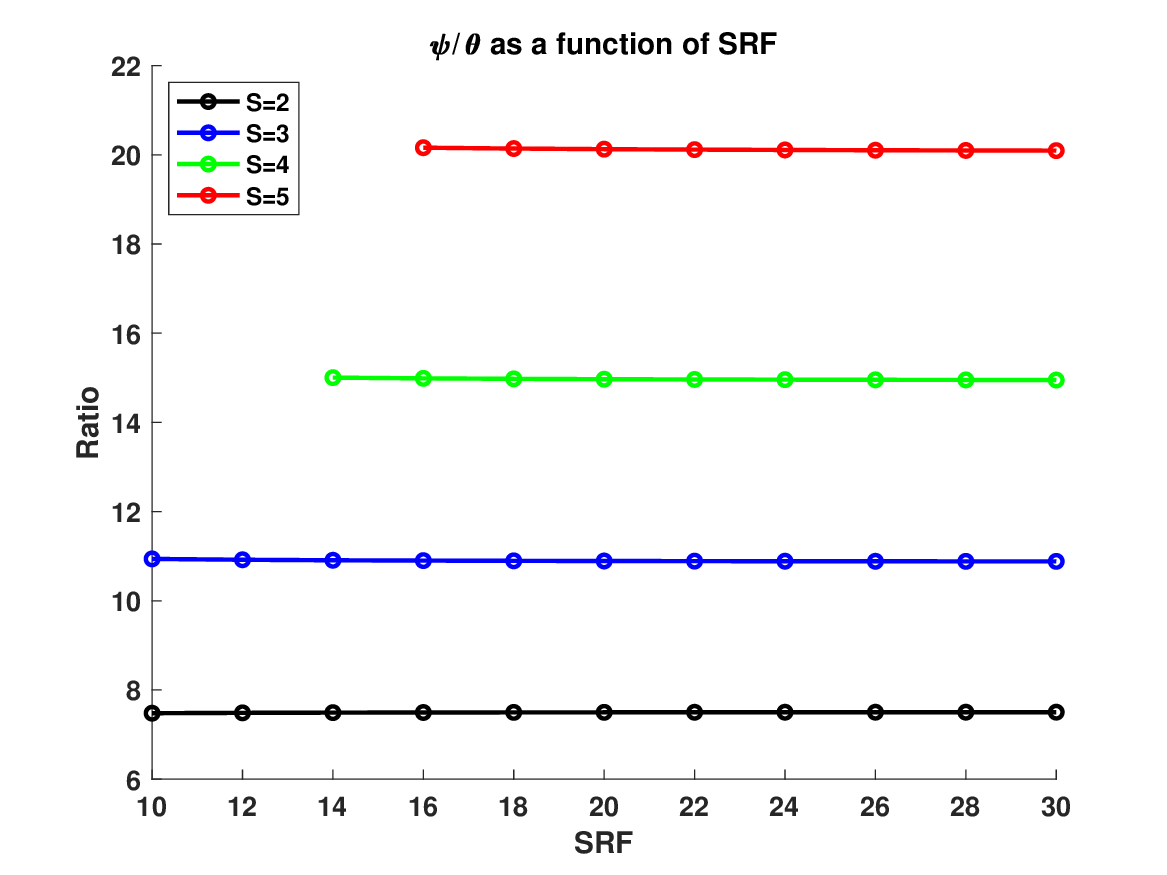}}
	\caption{The left figure displays $\theta(M,N,S)$ and $\psi(M,N,S)$ as functions of $\SRF=N/M$, while the right figure displays their ratio. For both experiments, we consider the range of parameters $2\leq S \leq 5$, $\pi S\leq \SRF$ and $10\leq\SRF\leq 30$, respectively.}
	\label{fig:exper2}
\end{figure}


\subsection{Related work on super-resolution limit}

A min-max formulation of super-resolution for measures on the grid can be traced back to \cite{donoho1992superresolution} and related works \cite{demanet2015recoverability,batenkov2020conditioning}. The similarity is that all three consider \textit{continuous} Fourier measurements and the main differences are their assumptions on $\mu$. In \cite{donoho1992superresolution}, $\mu$ is supported in a grid of the real line $\R$ with constraints on the ``density" of the support. In \cite{demanet2015recoverability}, $\mu$ is again supported on a grid in $\R$ but has $S$ atoms without restrictions on its support. In \cite{batenkov2020conditioning}, $\mu$ is supported a grid in $[0,1]$, has $S$ atoms, and its support satisfies a clumps condition. Our setup is different from those considered in \cite{donoho1992superresolution,demanet2015recoverability,batenkov2020conditioning} in the sense that we consider {\it discrete} Fourier measurements. Even though the min-max error in Theorem \ref{thm:minmax}, is  strikingly similar to those found in \cite{donoho1992superresolution,demanet2015recoverability,batenkov2020conditioning}, their results do not imply ours or vice versa.

Recently, the authors of \cite{batenkov2019super} studied the min-max error with continuous Fourier measurements, when the support set contains a {\it single cluster} of $\lambda$ nodes while the remaining $S-\lambda$ ones are well separated. In our notation, this corresponds to $\lambda_1=\lambda$ and $\lambda_a=1$ for all $a\geq 2$. Under appropriate conditions on the various parameters, \cite[Theorem 2.10]{batenkov2019super} shows that: (1) the min-max support error for the cluster nodes, as a function of $\SRF$, scales linearly in  ${\rm SRF}^{2\lambda -2}$; (2) the min-max amplitude error for the cluster nodes, scales linearly in  ${\rm SRF}^{2\lambda -1}$. Their result does not imply ours or vice versa, since we consider different models. Nonetheless, their min-max amplitude error is consistent with our Theorem \ref{thm:minmax} if the support set contains a cluster of $\lambda$ closely spaced points.

Regarding the detection of the number of sources, \cite{liu2020resolution} provided an information-theoretic condition for the SRF and the noise level such that the number of sources can be correctly detected.


\section{MUSIC and its super-resolution limit}
\label{sec:music}

In signal processing, a class of subspace methods, including MUSIC \cite{schmidt1986multiple}, is widely used due to their superior numerical performance. It was well known that MUSIC has super-resolution phenomenon \cite{odendaal1994two}. The resolution limit of MUSIC was discovered by numerical experiments in \cite{liao2016music}, but has never been rigorously proved. A main contribution of this paper is to prove the sensitivity of the noise-space correlation function in MUSIC under the clumps model.

\subsection{The MUSIC algorithm}

MUSIC is built upon a Hankel matrix and its Vandermonde decomposition. For a fixed positive integer $L \le M$, we form the Hankel matrix of $y$:
\[
\calH(y) 
:= 
\begin{bmatrix}
y_0 & y_1 & \hdots & y_{M-L}
\\
y_1 & y_2 & \hdots & y_{M-L+1}
\\
\vdots & \vdots & \ddots & \vdots
\\
y_{L} & y_{L+1} & \hdots & y_{M}
\end{bmatrix} \in \mathbb{C}^{(L+1) \times (M-L+1)}.
\]
Denote the noiseless measurement vector by $y^0 = \Phi(\Omega,M) x$. For simplicity, we will denote $\Phi(\Omega,M)$ by $\Phi_M$ in this section. It is straightforward to verify that $\calH(y^0)$ possesses the following Vandermonde decomposition:
$$
\calH(y^0)  
= \Phi_L X \Phi_{M-L}^T,
$$
where $X = \diag(x) \in \mathbb{R}^{S \times S}$. 
We always assume that $M+1 \ge 2S$ and $S \le L \le M-S+1$ so that $\Phi_L$ and $\Phi_{M-L}$ have full column rank, and $\calH(y^0)$ has rank $S$.

\renewcommand{\algorithmicrequire}{\textbf{Input:}}
\renewcommand{\algorithmicensure}{\textbf{Output:}}
\begin{algorithm}[t!]   
	\caption{MUltiple SIgnal Classification (MUSIC)}          	
	\label{algorithmmusic}		
	\begin{algorithmic}[1]      
		\REQUIRE  $y \in \C^{M+1}$, sparsity $S$ and a integer $L$
		\STATE Form Hankel matrix $\calH(y) \in \C^{(L+1)\times (M-L+1)}$
		\STATE Compute the SVD of $\calH(y)$: 
		\begin{equation}
		\label{svdnoisy}
		\calH(y) = [\underbrace{\hat U}_{(L+1) \times S} \  \underbrace{\hat W}_{(L+1) \times (L-S)}] \underbrace{{\rm diag}(\hat \sigma_1 , \ldots , \hat \sigma_S, \hat \sigma_{S+1} ,\ldots)}_{(L+1) \times (M-L+1)} [\underbrace{\hat V_1}_{(M-L+1) \times S} \ \underbrace{\hat V_2}_{(M-L+1) \times (M-L+1-S)}]^* \end{equation}
		where $\hat \sigma_1 \ge \hat \sigma_2  \ldots \ge  \hat \sigma_S \ge \hat \sigma_{S+1} \ge \ldots$ are the singular values of $\calH(y)$.
		\STATE Compute the imaging function $\hat\calJ(\omega) = {\|\phi_L(\omega)\|_2}/{ \|{\hat W}^* \phi_L (\omega)\|_2}, \ \omega \in [0,1)$.
		\ENSURE $\hat\Omega= \{\hat\omega_j\}_{j=1}^S$ corresponding to the $S$ largest local maxima of  $\hat\calJ$.
	\end{algorithmic}
	\label{algmusic}
\end{algorithm}

In the noiseless case, let the Singular Value Decomposition (SVD) of $\calH(y^0)$ be:
\begin{equation}
\calH(y^0) = [\underbrace{U}_{L \times S} \ \underbrace{W}_{L \times (L-S)}] \ \underbrace{\text{diag}(\sigma_1,\sigma_2,\ldots,\sigma_S,0,\ldots,0)}_{L \times (M-L+1)} \ [\underbrace{V_1}_{(M-L+1) \times S} \ \underbrace{V_2}_{(M-L+1) \times (M-L+1-S)}]^*
\label{svdnoiseless}
\end{equation}
where $\sigma_1 \ge \sigma_2 \ge \ldots \ge \sigma_S$ are the non-zero singular values of $\calH(y^0)$.  The column spaces of $U$ and $W$ are the same as $\text{Range}(\calH(y^0))$ and $\text{Range}(\calH(y^0))^\perp$, which are called the signal space and the noise space respectively.

For any $\omega \in [0,1)$ and integer $L$, we define the steering vector of length $L+1$ at $\omega$ to be 
\[
\phi_L(\omega) = [1 \ e^{-2\pi i \omega} \ e^{-2\pi i 2 \omega} \ \ldots \ e^{-2\pi i L \omega}]^T \in \mathbb{R}^{L+1},
\]
and then $\Phi_L = [\phi_L(\omega_1) \ \ldots \ \phi_L(\omega_S) ]$. The MUSIC algorithm is based on the following observation on the Vandermonde structure of $\Phi_L$: if $S\le L \le M-S+1$, then
\begin{equation}
\omega \in \{\omega_j\}_{j=1}^S
\Longleftrightarrow
\phi_L(\omega) \in \text{Range}(\Phi_L) 
= \text{Range}(\calH(y^0)) = \text{Range}(U).
\label{subspacebasis}
\end{equation}

\begin{table}[t]
	\renewcommand{\arraystretch}{1}
	\centering
	\resizebox{1\columnwidth}{!}{
		\begin{tabular}{ |c | c  |c| }
			\hline
			&   Noise-space correlation function&  Imaging function\\
			[5pt]
			\hline
			Noiseless case
			& $\calR(\omega) = \frac{\| W^* \phi_L(\omega)\|_2}{\|\phi_L(\omega) \|_2}$
			&   $\calJ(\omega) = \frac{1}{\calR(\omega)} = \frac{\|\phi_L(\omega) \|_2}{\| W^* \phi_L(\omega)\|_2}$
			\\
			[5pt]
			\hline
			Noisy case   & 
			$\hat\calR(\omega) = \frac{\| {\hat W}^* \phi_L(\omega)\|_2}{\|\phi_L(\omega) \|_2}$
			& 
			$\hat\calJ(\omega) = \frac{1}{\hat\calR(\omega)} = \frac{\|\phi_L(\omega) \|_2}{\| {\hat W}^* \phi_L(\omega)\|_2}$
			\\
			[5pt]
			\hline
		\end{tabular}
	}
	\caption{Noise-space correlation functions and imaging functions in MUSIC, where $W$ given in \eqref{svdnoiseless} and $\hat W$ given in \eqref{svdnoisy} of Algorithm \ref{algorithmmusic} are the noise space of dimension $L-S$ in the noiseless and noisy case respectively.}
	\label{TableMUSICFunctions}
\end{table}

We define a noise-space correlation function $\calR(\omega)$ and an imaging function as its reciprocal (see Table \ref{TableMUSICFunctions} for definitions). The following lemma is based on \eqref{subspacebasis}.

\begin{lemma}
	Suppose $M+1 \ge 2S$ and $L$ is chosen such that $S\le L \le M-S+1$. Then 
	\[
	\omega \in \{\omega_j\}_{j=1}^S
	\Longleftrightarrow
	\calR(\omega) = 0
	\Longleftrightarrow
	\calJ(\omega) = \infty.
	\]
\end{lemma}

In the noiseless case, the source locations can be exactly identified through the zeros of the noise-space correlation function $\calR(\omega)$
or the peaks of the imaging function $\calJ(\omega)$, as long as the number of measurements is at least twice the number of point sources to be recovered. 

In the presence of noise, $\calH(y^0)$ is perturbed to $\calH(y)$ such that:
$$\calH(y) = \calH(y^0) + \calH(\eta)$$
whose SVD is given by \eqref{svdnoisy} in Algorithm \ref{algorithmmusic}. 
The noise-space correlation function and the imaging function are perturbed to 
$\hat\calR$ and
$\hat\calJ$
respectively (see Table \ref{TableMUSICFunctions} for the definitions). MUSIC in Algorithm \ref{algmusic} gives rise to a recovered support $\hat\Omega= \{\hat\omega_j\}_{j=1}^S$ corresponding to the $S$ largest local maxima of  $\hat\calJ$.

\subsection{Sensitivity of the noise-space correlation function in MUSIC}
When the noise-to-signal ratio is low, the imaging function $\hat\calJ$ still peaks around the true sources, but MUSIC can fail when the noise-to-signal ratio increases.
The stability of MUSIC depends on the perturbation of the noise-space correlation function from $\calR$ to $\hat{\calR}$ which is measured by $\|\hat\calR - \calR\|_\infty := \max_{\omega \in [0,1)} |\hat\calR(\omega) - \calR(\omega)| $.
Thanks to the classical perturbation theory on singular subspaces by Wedin \cite[Section 3]{wedin1972perturbation}, we have the following bound for $\|\hat\calR - \calR\|_\infty$:
\begin{proposition}
	\label{lemmapermusic1}
	Suppose $M+1 \ge 2S$ and $L$ is chosen such that $S\le L \le M-S+1$. If $2\|\calH(\eta)\|_2 < x_{\min}\sigma_{\min}(\Phi_L)\sigma_{\min}(\Phi_{M-L})$, then
	\begin{equation}
		\|\hat\calR - \calR\|_\infty 
	\le \frac{2\|\calH(\eta)\|_2}{x_{\min}\sigma_{\min}(\Phi_L)\sigma_{\min}(\Phi_{M-L})}.	
	\label{eqrper}
	\end{equation}
\end{proposition}

The dependence on $\sigma_{\min}(\Phi_L)$ and $\sigma_{\min}(\Phi_{M-L})$ in \eqref{eqrper} are crucial since they are small in the super-resolution regime.
It is the best to set $ L=\lfloor M/2\rfloor$ to balance them. 
With Wedin's theorem, Proposition \ref{lemmapermusic1} improves Theorem 3 in \cite{liao2016music}, which upper bound $\|\hat\calR - \calR\|_\infty$ in the order of $\|\calH(\eta)\|_2/[\sigma^2_{\min}(\Phi_L)\sigma^2_{\min}(\Phi_{M-L})].$ Combining Proposition \ref{lemmapermusic1}  and Theorem \ref{thm:clump2} gives rise to the following sensitivity bound of the noise-space correlation function under the clumps model:

\begin{theorem}
	\label{thmmusic}
	Fix positive integers $A,M,S$. Let $M\geq S^2$ be an even integer and $L=M/2$. Suppose $\Omega$ satisfies Assumption \ref{def:clumps2} with parameters $(A,L,S,\beta,\alpha)$ for some $\alpha>0$ and $\beta$ with \eqref{eq:sep2}. 
	For each $1\leq a\leq A$, let $c_a = C_a(\lambda_a,L)$. If $ {4\|\calH(\eta)\|_2 \sum_{a=1}^A (c_a  \alpha^{-\lambda_a+1})^2 }<{x_{\min} M}$, then
	$$\|\hat\calR - \calR\|_\infty
	\le \frac{4\|\calH(\eta)\|_2}{x_{\min} M} \sum_{a=1}^A \big(c_a \, \alpha^{-\lambda_a+1} \big)^2.
	$$
\end{theorem}

\begin{remark}
In this paper, the stability of MUSIC is given in terms of the $L^\infty$ perturbation of the noise-space correlation function. When the point sources are well separated, a source localization error $|\widehat\omega_j -\omega_j |$ was derived in terms of   $\|\hat\calR - \calR\|_\infty$ \cite[Theorem 4]{liao2016music}. In the clumps model, such a derivation is more complicated and we leave it as a future work. 
\end{remark}

Theorem \ref{thmmusic} is stated for even $M$ only for simplicity, and a similar result holds for odd $M$ by setting $L= \lfloor M/2\rfloor$. This theorem applies to any noise vector $\eta$. In the case of bounded noise, applying the inequality $\|\calH(\eta)\|_2\leq \|\calH(\eta)\|_F\leq \sqrt L \|\eta\|_2$ gives rise to the following corollary:
\begin{corollary}
\label{comusicbounded}
Under the assumptions in Theorem \ref{thmmusic},
$$\|\hat\calR - \calR\|_\infty
	\le \frac{4\sqrt L \|\eta\|_2}{x_{\min} M} \sum_{a=1}^A \big(c_a \, \alpha^{-\lambda_a+1} \big)^2.
	$$ 
\end{corollary}

If $\eta$ is independent gaussian noise, i.e., $\eta \sim \mathcal{N}(0,\sigma^2 I)$, the spectral norm of $\calH(\eta)$ has been well studied in literature \cite{meckes2007spectral,adamczak2010few}. The following lemma \cite[Theorem 4]{liao2015multi} is obtained from the matrix Bernstein inequality \cite{tropp2012user}:
\begin{lemma}
\label{lemmahankelg}
	If $\eta \sim \mathcal{N}(0,\sigma^2 I)$, then
	\label{lemmanoise}
	\begin{align*}
	\mathbb{P}\left\{ 
	\|\calH(\eta)\|_2 \ge t
	\right\}
	&\le (M+2) \exp\left(-\frac{t^2}{2\sigma^2 \max(L+1,M-L+1)}\right), \forall t>0.
	\end{align*}  
\end{lemma}

Combining Theorem \ref{thmmusic} and Lemma \ref{lemmahankelg} gives rise to the following explicit bound on the noise-to-signal ratio $\sigma/x_{\min}$ in order to guarantee a fixed $\epsilon$-perturbation of $\calR$.

\begin{corollary}
\label{comusicgaussian}
Suppose $\eta \sim \mathcal{N}(0,\sigma^2 I)$. Fix $\epsilon>0$ and $\nu>1$.
Under the assumptions in Theorem \ref{thmmusic}, if
	\begin{equation}
	\frac{\sigma}{x_{\min}} <\frac{M}{4 \sqrt{\nu(M+2)\log(M+2)}}  \left(\sum_{a=1}^A \big(c_a  \alpha^{-\lambda_a+1} \big)^2 \right)^{-1} \epsilon,
	\label{thmmusiceq1}
	\end{equation}
	then
	$	\|\hat\calR - \calR\|_\infty 
	\le \epsilon 
	$
	with probability no less than $1-(M+2)^{-(\nu-1)}$.
\end{corollary}

Corollary \ref{comusicgaussian} is proved in Appendix \ref{secproofthmmusic}. It upper bounds the noise-to-signal ratio for which MUSIC can guarantee an $\epsilon$-perturbation of $\calR$. In the special case where each $\Lambda_a$ contains $\lambda$ equally spaced points with spacing $\alpha/M$ (see Figure \ref{FigDemoClumps1}), \eqref{thmmusiceq1} can be simplified to
\begin{equation}
\frac{\sigma}{x_{\min}} \propto\epsilon \sqrt{\frac{M}{\log M}} \alpha^{2\lambda-2}
=\epsilon \sqrt{\frac{M}{\log M}} \left(\frac{1}{\rm SRF}\right)^{2\lambda-2},
\label{musicres2}
\end{equation}
which shows that the noise-to-signal ratio that MUSIC can tolerate is exponential in $1/{\rm SRF}$. The key contribution of this paper is that, the exponent only depends on the cardinality of the clumps instead of the total sparsity $S$. These estimates are verified by numerical experiments in Section \ref{subsecmusicnum}.

\subsection{Numerical simulations on the super-resolution limit of MUSIC}
\label{subsecmusicnum}

In our experiments, the true support $\Omega$ consists of $A=1,2,3,4$ clumps and each clump contains $\lambda$ equally spaced point sources separated by $\Delta = \alpha/M$ where $1/\alpha$ is the SRF of $\Omega$ (see Figure \ref{Fig_TransitionImage} (a) for an example). The clumps are separated at least by $10/M$. The $x_i$'s are complex with unit magnitudes and random phases. Noise is gaussian: $\eta\sim \mathcal{N}(0,\sigma^2 I)$.
We set $M=100$ and let $\Delta$ vary so that the $\OmegaRF$ varies from $1$ to $10$. 
We run MUSIC with the varying $\OmegaRF$ and $\sigma$ for $10$ trials. The support error is measured by the matching distance between $\Omega$ and $\hat\Omega$:
$${\rm dist}_B(\Omega,\hat\Omega) := \inf_{\substack{\text{bijection }  \psi:\  \Omega\rightarrow  \hat\Omega } } \  \sup_{\hat\omega \in \hat\Omega} 
\left |\hat\omega - \psi(\omega)\right|_{\mathbb{T}}
.$$ 
Figure \ref{Fig_TransitionImage} (b) displays the average $\log_2 [{{\rm dist}_B(\Omega,\hat\Omega)}/{\Delta}]$ over $10$ trials with respect to $\log_{10}\OmegaRF$ (x-axis) and $\log_{10}\sigma$ (y-axis) when $\Omega$ contains $2$ clumps of $3$ equally spaced point sources: $A=2, \lambda = 3$. A clear phase transition demonstrates that MUSIC is capable of resolving closely spaced \textit{complex-valued} objects as long as $\sigma$ is below certain threshold. 

\begin{figure}[h]
	\subfigure[$2$ clumps of $3$ equally spaced point sources separated by $\Delta = \alpha/M$]
	{
		\begin{tikzpicture}[xscale = 1.5,yscale = 1.5]
		\draw[thick] (-5.5,0) -- (-0.1,0);
		\draw[red,thick,->] (-5,0) -- (-5,2);
		\draw[red,thick,->] (-4.7,0) -- (-4.7,2);
		\draw[red,thick,->] (-4.4,0) -- (-4.4,2);
		\draw[blue,thick,<->] (-4.7,-0.2) -- (-4.4,-0.2);
		\node[blue,below] at (-4.1,-0.2) {$ \alpha/M$};
		\draw[blue,thick,<->] (-4.4,1) -- (-1,1);
		\node[blue,above] at (-2.8,1) {$\ge {\rm Constant}/M$};
		\draw[red,thick,->] (-1,0) -- (-1,2);
		\draw[red,thick,->] (-0.7,0) -- (-0.7,2);
		\draw[red,thick,->] (-0.4,0) -- (-0.4,2);
		\draw[blue,thick,<->] (-1,-0.2) -- (-0.7,-0.2);
		\node[blue,below] at (-0.7,-0.2) {$ \alpha/M$};
		\node[above] at (-4.6,2.2) {$\Lambda_1$};
		\node[above] at (-0.75,2.2) {$\Lambda_2$};
		\node[below] at (-2.5,-1) {$\lambda = \max(|\Lambda_1|,|\Lambda_2|) = 3$};
		\end{tikzpicture}
	}
	\subfigure[Phase transition of MUSIC]{
		\includegraphics[width=0.55\textwidth]{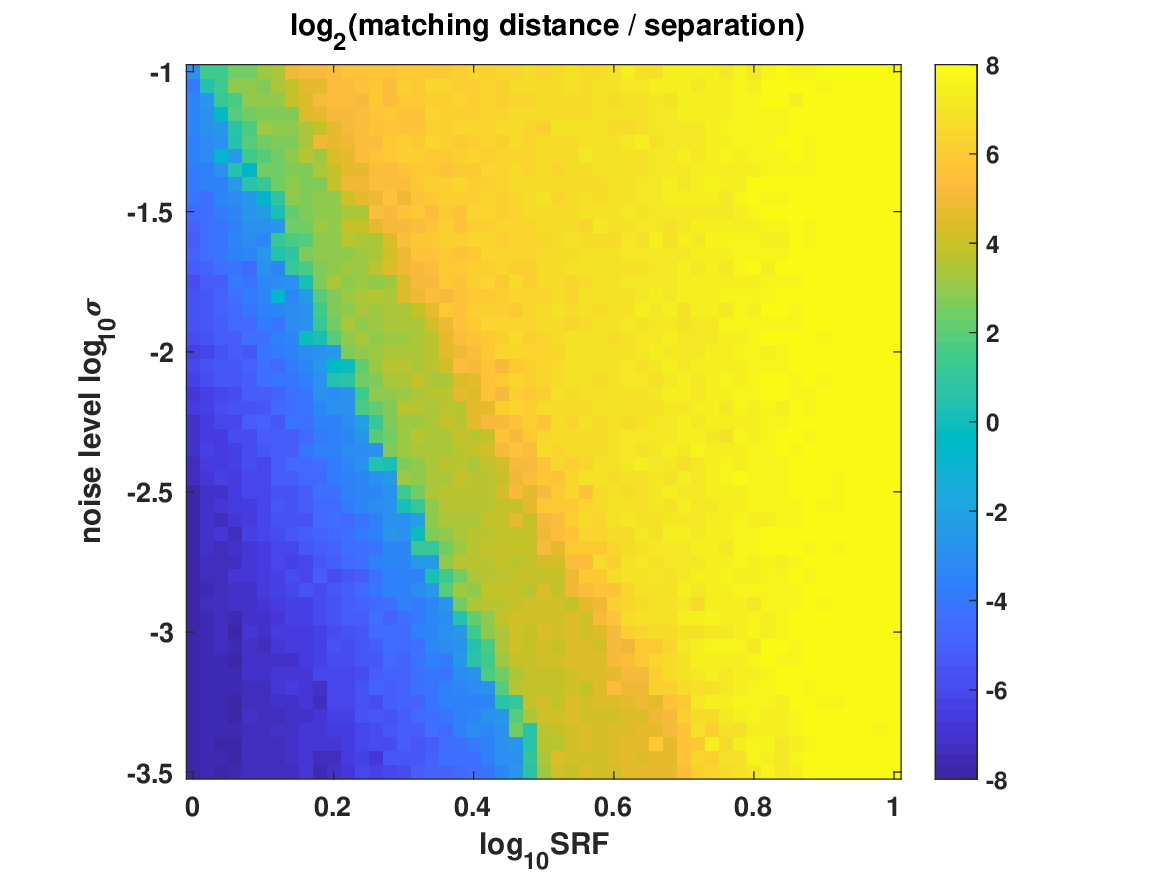}
	}
	\caption{Figure \ref{Fig_TransitionImage} (b) displays the average $\log_2 [{{\rm dist}_B(\Omega,\hat\Omega)}/{\Delta}]$ over $10$ trials with respect to $\log_{10}\OmegaRF$ (x-axis) and $\log_{10}\sigma$ (y-axis) when $A=2$ and $\lambda =3$.}
	\label{Fig_TransitionImage}
\end{figure}

\quad 

\begin{table}[h]
	\centering
	\begin{tabular}{ | c | c | c | c | c | c|  c|}
		\hline			
		& $\lambda = 2$ & $\lambda = 3$ & $\lambda = 4$ & $\lambda = 5$ & Numerical $q(\Omega)$ &Theoretical $q(\Omega)$  \\
		\hline
		1-clump: $A=1$  & $3.0019$ & $5.1935$ & $7.4176$ & $10.4286$ & $2.45\lambda-2.07$ & $2\lambda-2$ \\
		\hline
		2-clump: $A=2$  & $3.1287$ & $5.2717$ & $7.9371$ & $10$ & $2.33\lambda-1.56$ & $2\lambda-2$ \\
		\hline
		3-clump: $A=3$  & $3.1081$ & $5.1826$ & $7.299$ & $8.5$ & $1.83\lambda-0.38$ & $2\lambda-2$ \\
		\hline
		4-clump: $A=4$  & $3.0767$ & $5.1731$ & $7.3252$ & $10$ &$2.29\lambda-1.63$  & $2\lambda-2$ \\
		\hline  
	\end{tabular}
	\caption{Numerical simulations of $q(\Omega)$ in \eqref{eqmusicreslimit} on the super-resolution limit of MUSIC.}
	\label{TableMUSIC}
\end{table}

In Figure \ref{Fig_PhaseTransition}, we display the phase transition curves at which ${\rm dist}_B(\Omega,\hat\Omega) \approx {\Delta}/2$  with respect to $\log_{10}\OmegaRF$ (x-axis) and $\log_{10}\sigma$ (y-axis) when $\Omega$ contains $A=1,2,3,4$ clumps of $\lambda = 2,3,4,5$ equally spaced point sources. All phase transition curves are almost straight lines, manifesting that the noise level $\sigma$ that MUSIC can tolerate satisfies 
\begin{equation}
\sigma \propto \OmegaRF^{-q(\Omega)}.
\label{eqmusicreslimit}
\end{equation}
A least squares fitting of the curves by straight lines gives rise to the exponent $q(\Omega)$ numerically, summarized in Table \ref{TableMUSIC}. It is close to our theory \eqref{musicres2} of $q(\Omega) = 2\lambda -2$.

\begin{figure}[h]
	\centering 
	\subfigure[1-clump model: $A=1$]{
		\includegraphics[width=0.45\textwidth]{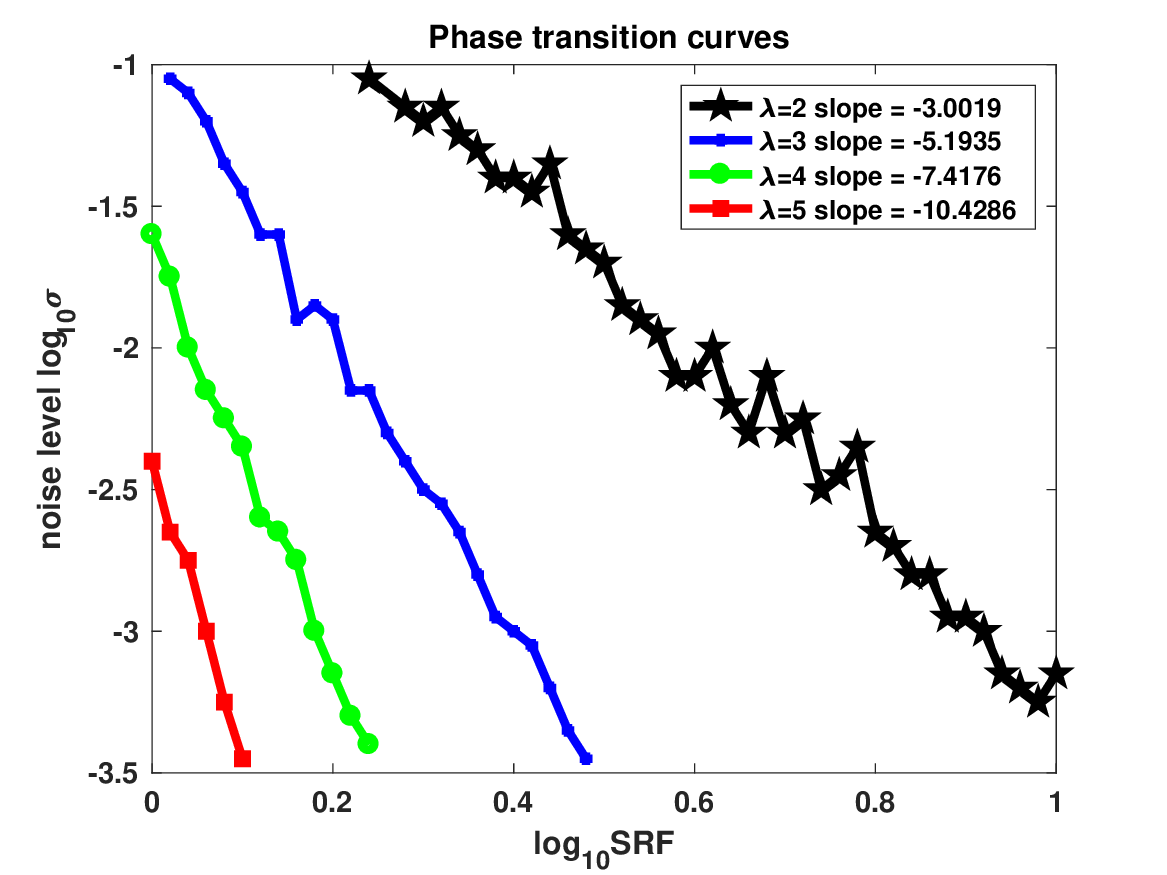}
	}
	\subfigure[2-clump model: $A=2$]{
		\includegraphics[width=0.45\textwidth]{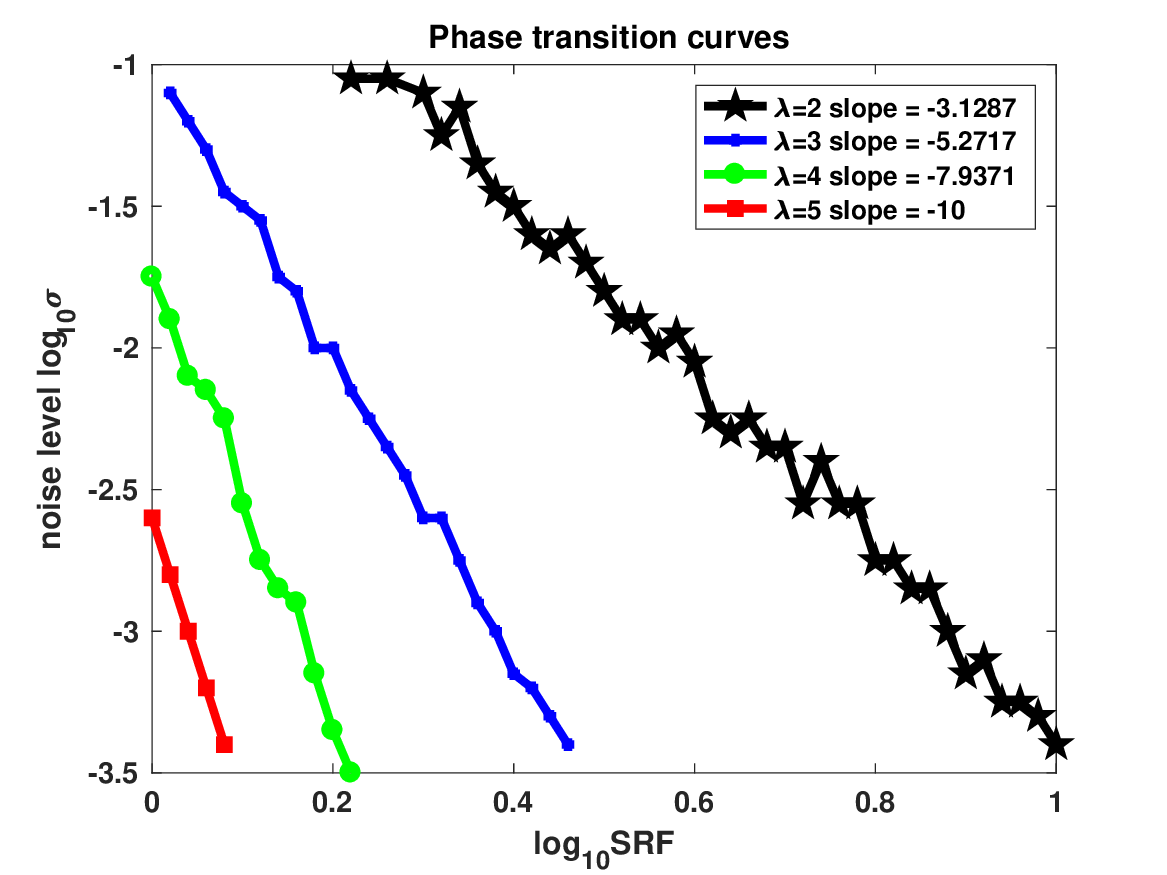}
	}
	\subfigure[3-clump model: $A=3$]{
		\includegraphics[width=0.45\textwidth]{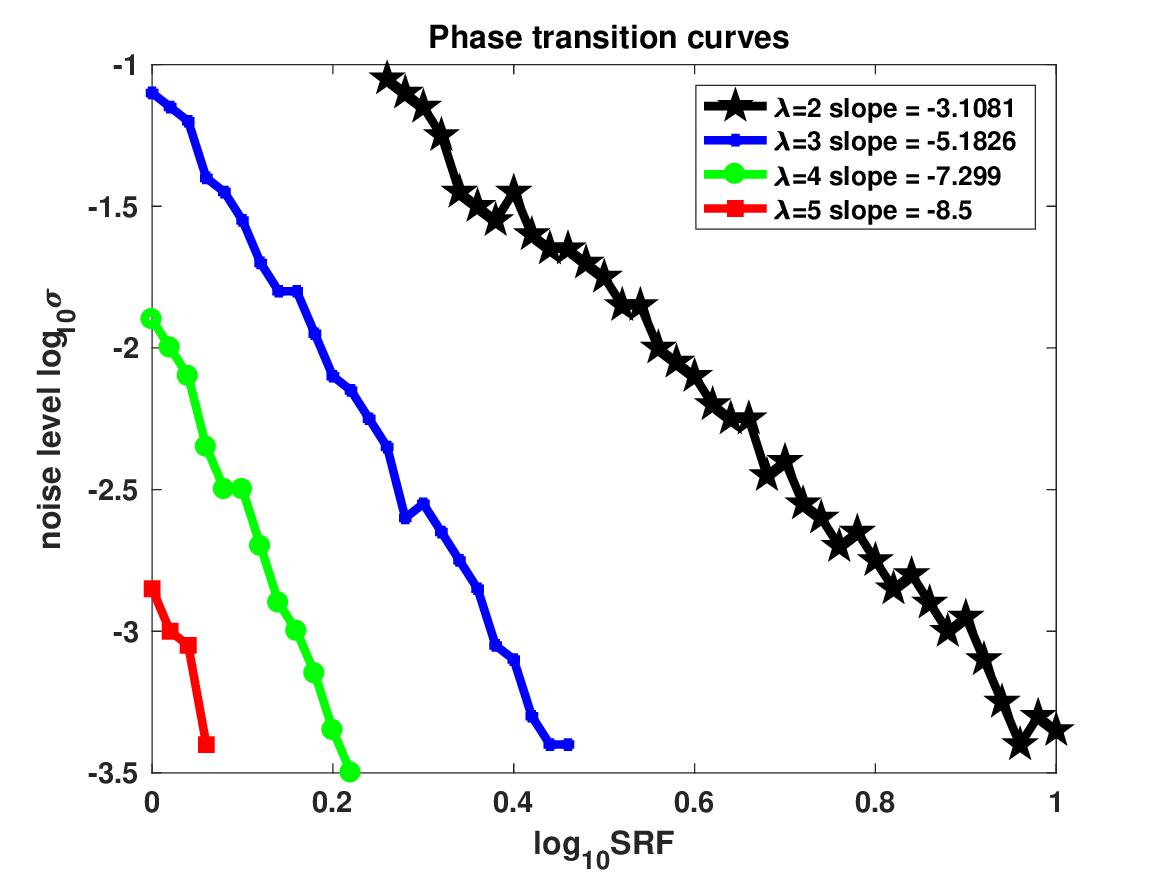}
	}
	\subfigure[4-clump model: $A=4$]{
		\includegraphics[width=0.45\textwidth]{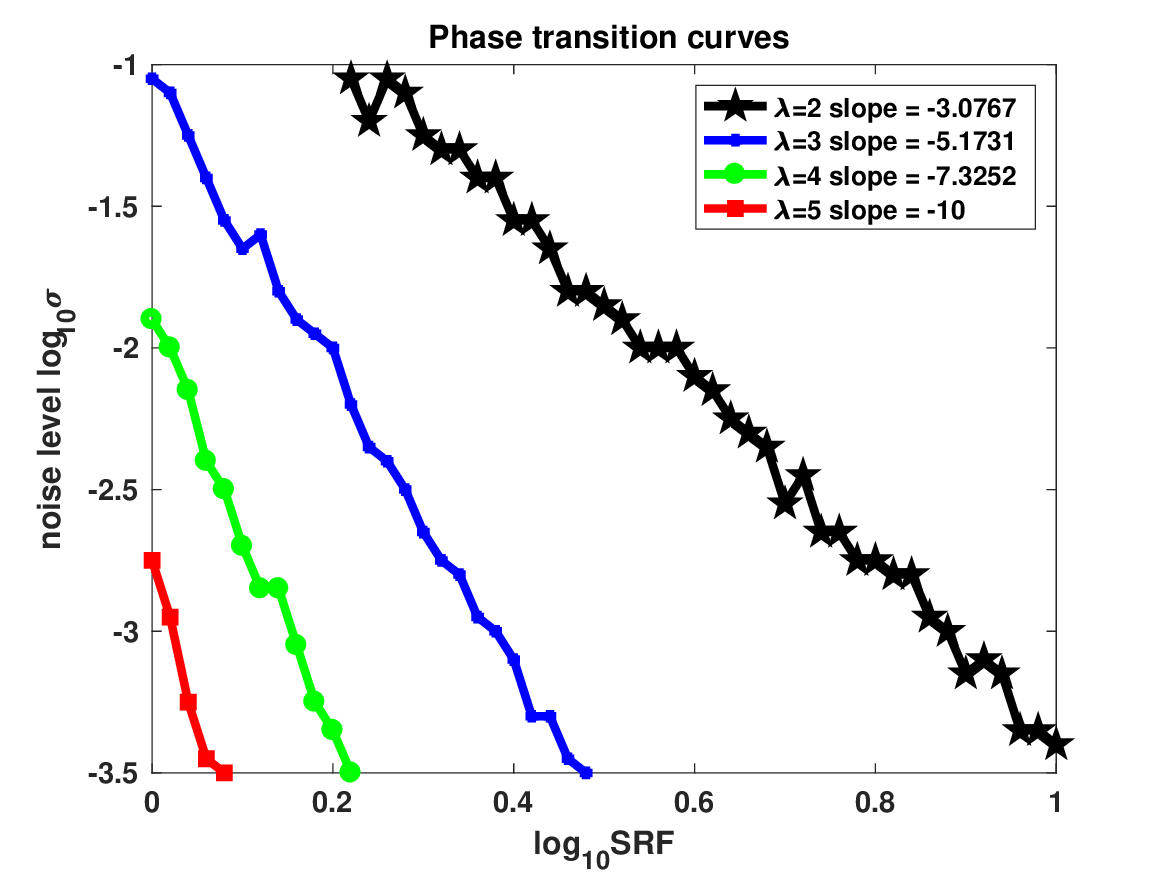}
	}
	\caption{Phase transition curves: (a-d) shows the average $\log_2 [{{\rm dist}_B(\Omega,\hat\Omega)}/{\Delta}]$ over $10$ trials with respect to $\log_{10}\OmegaRF$ (x-axis) and $\log_{10}\sigma$ (y-axis) when $\Omega$ contains $A=1,2,3,4$ clumps of $\lambda = 2,3,4,5$ equally spaced point sources. The slope is computed by a least squares fitting by a straight line.}
	\label{Fig_PhaseTransition}
\end{figure}

\subsection{Related work on subspace methods}

Subspace methods, including MUSIC \cite{schmidt1986multiple}, ESPRIT \cite{kailath1989esprit} and MPM \cite{hua1990matrixpencil}, were initially proposed for the Direction-Of-Arrival (DOA) estimation \cite{krim1996array}. In the DOA setting, the amplitude vector $x$ is random with respect to time, and multiple snapshots of measurements are taken. The covariance matrix of $y$ possesses a Vandermonde decomposition so that MUSIC, ESPRIT and MPM are applicable. The classical theories in  \cite{stoica1989music,ottersten1991performance} primarily analyze the stability of MUSIC and ESPRIT with respect to the number of snapshots, denoted by $\#{\rm Snapshot}$. They show that the asymptotic distribution of the squared error is on the order of $C\times {\rm noise}/\#{\rm Snapshot}$ where the constant $C$ depends on $\sigma_{\min}(\Phi)$. It is not clear from \cite{stoica1989music,ottersten1991performance} how large the implicit constant $C$ is, while our paper gives an explicit characterization of the $C$ under the clumps model.

The super-resolution problem considered in this paper corresponds to the single snapshot case, in which case, the key quantity of interest is the implicit constant $C$. In the single-snapshot case, there have been several works on the stability of subspace methods, which addressed the connection between $\sigma_{\min}(\Phi)$ and the support $\Omega$. These works include MUSIC \cite{liao2016music}, ESPRIT \cite{fannjiang2016compressive} and MPM \cite{moitra2015matrixpencil}, but only apply to the well-separated case. During the review period of this paper, the present authors applied Theorem \ref{thm:clump2} to derive an error bound for ESPRIT in the super-resolution regime \cite{li2020super}.

In literature, the resolution of MUSIC  was addressed in  \cite{stoica1995resolution,lee1992cramer,lee1992eigenvalues}.  These papers studied the frequency estimation error in the DOA setting with multiple snapshots of measurements. The Cram\'er-Rao lower bound is derived based on the statistics of random amplitudes. Thus, these results do not apply to the single-snapshot case. It was explicitly mentioned in \cite[after Equation (1.2)]{stoica1995resolution} that, \lq\lq In effect, the analysis in \cite{lee1992cramer} assumes the multi-experiment case, and hence may not apply to the single experiment case.\rq\rq\  Our theory shares some similarity with \cite{lee1992cramer}. Equation (2) in \cite{lee1992cramer} implies that, to guarantee the success of MUSIC for the estimation of $\lambda$ closely spaced frequencies with separation $\Delta$, the SNR needs to be at least $\Delta^{-2(\lambda-1)}$, which has the same dependence on $\Delta$ as the Equation (4.6) of our paper. Here are some differences between Equation (2) in \cite{lee1992cramer} and Equation (4.6) of our paper: (1) Our Equation (4.6) implies SNR in the order of ${\rm SRF}^{2(\lambda-1)} = (M\Delta)^{-2(\lambda-1)}$, which is more accurate than  $\Delta^{-2(\lambda-1)}$. In other words, our theory takes the advantage of a large $M$ in applications. (2) Our theory does not rely on the randomness of $x$, and considers the clumps model, which is more general than just having  $\lambda$ closely spaced frequencies with separation $\Delta$. (3) Our results are non-asymptotic and work for finite measurements $M$, while the theories in \cite{stoica1995resolution,lee1992cramer,lee1992eigenvalues} are asymptotic and require $\#{\rm Snapshot} \rightarrow \infty$ and $\Delta\rightarrow 0$. 

Finally we remark that the results in the present paper are not the same as the Cram\'er-Rao bounds for a single-snapshot case. In the single snapshot case, one assumes certain statistics of the random noise, and the Cram\'er-Rao bound expresses  a lower bound on the variance of unbiased estimators. The noise statistics play an important role in the the Cram\'er-Rao bound. Our perturbation bound for MUSIC, such as Corollary \ref{comusicbounded}, is deterministic and can be applied for all noise vectors, independently of the noise statistics.


\commentout{
\section{Conclusion}

Without any additional assumptions on the unknown measure, super-resolution is an unstable inverse problem. Past works showed that the assumption $\Delta>C/M$, for a sufficiently large universal constant $C$, regularizes the problem and there are several algorithms that provably succeed in this regime. However, the $\Delta<1/M$ situation is important in super-resolution imaging and a different kind of assumption is needed. As seen in our estimate for the min-max error, a sparsity assumption is typically not strong enough unless the noise energy is significantly smaller than $\SRF^{-2S+1}$. In this paper, we regularized the inverse problem by imposing a geometric constraint on $\Omega$. 
More specifically, if $\Omega$ consists of well separated clumps, then MUSIC is stable provided that the noise energy is significantly smaller than $\SRF^{-2\lambda+1}$. This is the first result to rigorously confirm prior numerical evidence that MUSIC can succeed in the regime $\Delta<1/M$. More generally, this paper highlights that super-resolution recovery is possible under more situations than previously indicated by sparsity models.
}

\section*{Acknowledgements}

Both authors express their sincere gratitude to Professors John Benedetto, Sinan G\"unt\"urk, and Kasso Okoudjou for valuable advice and insightful discussions. Weilin Li was supported by the Defense Threat Reduction Agency Grant 1-13-1-0015, and the Army Research Office Grants W911 NF-17-1-0014 and W911 NF-16-1-0008. He was also supported by the National Science Foundation Grant DMS-1440140 during the Spring 2017 semester while in residence at the Mathematical Sciences Research Institute in Berkeley, California, and by the University of Maryland's Ann G. Wylie Dissertation fellowship during the Spring 2018 semester. Wenjing Liao was supported by NSF DMS 1818751, a startup fund from Georgia Institute of Technology, and the AMS-Simons Travel Grants.


\begin{appendices}
	
\section{Proof of theorems and corollaries}
\label{sec:proofs}

\subsection{Proof of Theorem \ref{thm:clump1}}

\label{proof:clump1}

We first need to introduce the following function that is of great importance in Fourier analysis. For a positive integer $P$, we define the normalized {\it Fej\'er kernel} $F_P\in C^\infty(\T)$ by the formula,
\[
F_P(\omega)
:=\frac{1}{P+1}\sum_{m=-P}^P \(1-\frac{|m|}{P+1}\) e^{2\pi im\omega}
=\frac{1}{(P+1)^2} \(\frac{\sin(\pi (P+1)\omega)}{\sin(\pi \omega)}\)^2. 
\]
The normalization is chosen so that $F_P(0)=1$. We recall some basic facts about the Fej\'er kernel. Its $L^2(\T)$ norm can be calculated using Parseval's formula, and so
\begin{equation}
\label{eq:fejer1}
\|F_P\|_{L^2(\T)}
=\frac{1}{(P+1)^2}\((P+1)^2+2\sum_{m=1}^P m^2\)^{1/2}
\leq \frac{1}{(P+1)^{1/2}}. 
\end{equation}
We can also provide a point-wise estimate. By the trigonometric inequality $|\sin(\pi \omega)|\geq 2|\omega|_\T$, we have
\begin{equation}
\label{eq:fejer2}
|F_P(\omega)|\leq \frac{1}{2^2(P+1)^2|\omega|_\T^2},
\quad\text{for all}\quad
\omega\in \T.
\end{equation}
If we raise the Fej\'er kernel to a power $R$, then the function $(F_P(\omega))^R$ has better decay, but at the cost of increasing its frequency support. If we keep the product $PR$ fixed, then increasing $R$ leads to better decay at the expense of worse localization near the origin.

The proof of Theorem \ref{thm:clump1} relies on the quantitative properties of a set of polynomials $\{I_j\}_{j=1}^S$, with $I_j$ depending on $\Omega$ and $M$, which we shall explicitly construct. The construction seems complicated, but the idea is very simple. For each $\omega_j\in\Omega$, we construct $I_j\in\P(M)$ such that it decays rapidly away from $\omega_j$ and
\[
I_j(\omega_k)=\delta_{j,k},
\quad\text{for all}\quad
\omega_k\in \Lambda_a.
\]
The key is to carefully construct each $I_j$ so that it has small norm; otherwise, the resulting lower bound for $\sigma_{\min}(\Phi)$ would be loose and have limited applicability. The construction of these polynomials is technical and it can be found in Appendix \ref{sec:lemmas}.

\begin{lemma}
	\label{lemma:localizedpoly}
	Suppose the assumptions of Theorem \ref{thm:clump1} hold and define the constant 
	\[
	\tilde B_a
	:=\tilde B_a(\lambda_a,M)
	:= \(1-\frac{\pi^2}{3\lambda_a^2}\)^{-(\lambda_a-1)/2} \(\frac{M}{\lambda_a}\)^{\lambda_a-1} \Big\lfloor \frac{M}{\lambda_a}\Big\rfloor^{-(\lambda_a-1)}. 
	\]
	For each $1\leq a\leq A$ and each $\omega_j\in\Lambda_a$, there exists a $I_j\in\P(M)$ satisfying the following properties.
	\begin{enumerate}[(a)]
		\item 
		$I_j(\omega_k)=\delta_{j,k}$ for all $\omega_k\in\Lambda_a$.
		\item 
		$|I_j(\omega_k)|\leq 1/(20S)$ for all $\omega_k\not\in\Lambda_a$.
		\item 
		$\|I_j\|_{L^2(\T)}
		\leq (2/M)^{1/2} \tilde B_a \lambda_a^{\lambda_a} \rho_j.$
	\end{enumerate}
\end{lemma}

\begin{proof}[Proof of Theorem \ref{thm:clump1}]
	Let $\{I_j\}_{j=1}^S$ be the polynomials constructed in Lemma \ref{lemma:localizedpoly}. Let $v\in\C^S$ be a unit norm vector such that 
	\[
	\sigma_{\min}(\Phi)
	=\|\Phi v\|_2.
	\]
	We define the trigonometric polynomial $I\in \P(M)$ by the formula,
	\[
	I(\omega)
	:=I(\omega,v)
	:=\sum_{j=1}^S v_j I_j(\omega).
	\]
	For each index $1\leq k\leq S$, we define the quantity
	\[
	\epsilon_k
	:=I(\omega_k)-v_k.
	\]
	Since $I_j(\omega_j)=1$, we have
	\[
	\epsilon_k
	=\sum_{j\not = k} v_j I_j(\omega_k).
	\]
	Fix a $\omega_k\in\Omega$. Then $\omega_k\in\Lambda_a$ for some $1\leq a\leq A$. By Cauchy Schwartz, the assumption that $v$ is unit norm, the property that $I_j(\omega_k)=\delta_{j,k}$ for all $\omega_k\in \Lambda_a$, and the upper bound on $|I(\omega_k)|$ given in Lemma \ref{lemma:localizedpoly}, we deduce
	\[
	|\epsilon_k|
	\leq \(\sum_{j\not=k} |I_j(\omega_k)|^2\)^{1/2}
	= \(\sum_{\omega_j\not\in\Lambda_a} |I_j(\omega_k)|^2\)^{1/2}
	\leq \frac{1}{20\sqrt{S}}. 
	\]
	This holds for each $1\leq k\leq S$, so we have 
	\[
	\|\epsilon\|_2\leq \sqrt{S} \|\epsilon\|_\infty \leq \frac{1}{20}.
	\]
	The conditions of robust duality, Proposition \ref{prop:duality2}, are satisfied, so we have
	\[
	\sigma_{\min}(\Phi)
	=\|\Phi v\|_2
	\geq \frac{19}{20}\|I(\cdot,v)\|_{L^2(\T)}^{-1}.
	\]
	To complete the proof, we need to upper bound $\|I(\cdot,v)\|_{L^2(\T)}$ uniformly in $v$. We use Cauchy-Schwartz, that $v$ has unit norm, and the norm bound for $I_j$ given in Lemma \ref{lemma:localizedpoly} to obtain the upper bound,
	\[
	\|I(\cdot,v)\|_{L^2(\T)}
	\leq \(\sum_{j=1}^S \|I_j\|_{L^2(\T)}^2 \)^{1/2}
	\leq \(\frac{2}{M}\)^{1/2} \(\sum_{j=1}^S (\tilde B_a \lambda_a^{\lambda_a}\rho_j)^2\)^{1/2},
	\] 
	where $\tilde B_a$ was defined in the referenced lemma. Combining the previous two inequalities shows that 
	\[
	\sigma_{\min}(\Phi)
	\geq \frac{19}{20\sqrt 2} \sqrt M  \(\sum_{j=1}^S (\tilde B_a^2\lambda_a^{\lambda_a}\rho_j)^2\)^{-1/2}
	=\sqrt M  \(\sum_{j=1}^S (B_a\lambda_a^{\lambda_a}\rho_j)^2\)^{-1/2},
	\]
	where by definition, $B_a = \tilde B_a(20\sqrt 2)/19$. 
\end{proof}

\subsection{Proof of Theorem \ref{thm:clump2}}

\label{proof:clump2}

\begin{proof}
	Fix an index $1\leq a\leq A$ and $\omega_j\in \Lambda_a$. Recalling the definition of $\rho_j$ and using that $\Delta\geq \alpha/M$, we see that
	\[
	\rho_j
	=\prod_{\omega_k\in \Lambda_a \setminus\{\omega_j\}}
	\frac{1}{\pi M|\omega_k-\omega_j|_\T}
	\leq \(\frac{1}{\pi\alpha}\)^{\lambda_a-1}. 
	\]
	This implies that 
	\[
	\frac{10\lambda_a^{5/2} (S\rho_j)^{1/(2\lambda_a)}}{M}
	\leq \frac{10\lambda_a^{5/2} S^{1/2}}{M\alpha^{1/2}}.
	\]
	This in turn, shows that the separation condition \eqref{eq:sep2} implies \eqref{eq:sep1}. Hence, the assumptions of Theorem \ref{thm:clump1} are satisfied, and we have
	\[
	\sigma_{\min}(\Phi)
	\geq \sqrt{M} \(\sum_{a=1}^A \sum_{\omega_j\in \Lambda_a} (B_a \lambda_a^{\lambda_a} \rho_j)^2 \)^{-1/2}. 
	\]
	We can write the right hand side in terms of $\alpha$. Observe that if $\tilde\Lambda_a=\{\tilde \omega_j\}_{j=1}^{\lambda_a}$ contains $\lambda_a$ points that are equispaced by $\alpha/M$ and $\tilde\rho_j$ is the complexity of $\tilde\omega_j$, then 
	\[
	\sum_{\omega_j\in \Lambda_a} \rho_j^2
	\leq \sum_{\tilde\omega_j\in\tilde\Lambda_a} \tilde\rho_j^2.
	\]
	Thus we have the inequality,
	\[
	\sum_{\omega_j\in \Lambda_a} \rho_j^2
	\leq \sum_{j=1}^{\lambda_a} \(\prod_{k=1,\ k\not=j}^{\lambda_a} \frac{1}{(j-k)^2}\) \(\frac{1}{\pi\alpha}\)^{2\lambda_a-2}.
	\]
	Combining the above inequalities completes the proof.
\end{proof}

\subsection{Proof of Theorem \ref{thm:theta}}

\label{proof:theta}

The crux of the proof is to construct, for each $\Omega\subset \{n/N\}_{n=0}^{N-1}$ with $|\Omega|=S$, a family of polynomials $\{H_j(\cdot,\Omega)\}_{j=1}^S$ with small $L^2(\T)$ norms that satisfy an appropriate interpolation property. The construction is technical because it must be done carefully in order to obtain an accurate bound for the lower restricted isometry constant. The proof of the following lemma can be found in Appendix \ref{sec:lemmas}. 
\begin{lemma}
	\label{lemma:poly1}
	Suppose the assumptions of Theorem \ref{thm:theta} hold and let $C(M,S)$ be the constant defined in the theorem. For each $\Omega\subset \{n/N\}_{n=0}^{N-1}$ and of cardinality $S$, there exist a family of polynomials $\{H_j(\cdot,\Omega)\}_{j=1}^S\subset\P(M)$ such that
	\[
	H_j(\omega_k,\Omega)=\delta_{j,k}
	\quad\text{for all}\quad \omega_j,\omega_k\in\Omega. 
	\]
	Moreover, we have the upper bound,
	\[
	\(\sum_{j=1}^S \|H_j(\cdot,\Omega)\|_{L^2(\T)}^2\)^{1/2}
	\leq C(M,S)^{-1} \frac{1}{\sqrt{M}} \(\frac{N}{M}\)^{S-1}. 
	\]
\end{lemma}

\begin{proof}[Proof of Theorem \ref{thm:theta}]
	By definition of the lower restricted isometry constant, there exists a set $\Omega$ of cardinality $S$ and supported on the grid with spacing $1/N$ such that
	\[
	\Theta(M,N,S)=\sigma_{\min}(\Phi(\Omega,M)). 
	\]
	Let $\{H_j(\Omega)\}_{j=1}^S$ be the family of polynomials given in Lemma \ref{lemma:poly1}. Let $u=u(\Omega)\in\C^S$ be a unit norm vector such that
	\[
	\sigma_{\min}(\Phi(\Omega,M))=\|\Phi(\Omega,M)u\|_2.
	\]
	We define the polynomial,
	\[
	H(\omega)
	:=H(\omega,u,\Omega)
	:=\sum_{j=1}^S u_j H_j(\omega,\Omega). 
	\]
	Using the interpolation property of $\{H_j(\cdot,\Omega)\}_{j=1}^S$ guaranteed by Lemma \ref{lemma:poly1}, we see that $H\in\P(\Omega,M,u)$. By exact duality, Proposition \ref{prop:duality}, we have
	\[
	\sigma_{\min}(\Phi(\Omega,M))
	=\max_{f\in \mathcal{P}(\Omega,M,u(\Omega))} \|f\|_{L^2(\T)}^{-1}
	\geq \|H(\cdot,\Omega)\|_{L^2(\T)}^{-1}. 
	\]
    Using Cauchy-Schwartz and that $u$ is a unit norm vector, we have
	\[
	\|H\|_{L^2(\T)}
	\leq \(\sum_{j=1}^S \|H_j\|_{L^2(\T)}^2\)^{1/2}.
	\]
	Combining the previous inequalities and using the upper bound given in Lemma \ref{lemma:poly1} completes the proof of the theorem.

\end{proof}

\subsection{Proof of Theorem \ref{thm:minmax}}

\label{proof:minmax}

\begin{proof}
	The upper bound for the min-max error is a direct consequence of Proposition \ref{prop:demanet} and Theorem \ref{thm:theta}. To obtain a lower bound for the min-max error, we first apply Proposition \ref{prop:upper} to the case that $\Omega$ consists of $2S$ consecutive points spaced by $1/N$. We ready check that the size assumptions on $M$ and $N$ imply that the conditions of Proposition \ref{prop:upper} are satisfied, and thus,
	\[
	\Theta(M,N,2S)
	\leq 2{4S-2 \choose 2S-1}^{-1/2} \sqrt{M+1}\ \(\frac{2\pi M}{N}\)^{2S-1}.
	\]
	Combining this with Proposition \ref{prop:demanet} establishes a lower bound for the min-max error.
\end{proof}

\subsection{Proof of Corollary \ref{comusicgaussian}}
\label{secproofthmmusic}

\begin{proof}
According to Theorem \ref{thmmusic}, a sufficient condition for $\|\hat\calR - \calR\|_\infty\le \epsilon$ is
 \begin{equation} \|\calH(\eta)\|_2 \le {M x_{\min}} \(\sum_{a=1}^A  (c_a  \alpha^{-\lambda_a+1})^2  \)^{-1}\epsilon/4.
	\label{thmmusicpeq1}
	\end{equation}
	Lemma \ref{lemmanoise} implies that \eqref{thmmusicpeq1} holds with probability no less than $1-(M+2)^{-(\nu-1)}$ as long as $t= {M x_{\min}} \(\sum_{a=1}^A  c^2_a  \alpha^{-2(\lambda_a-1)}  \)^{-1}\epsilon/4$ and $(M+2)e^{-\frac{t^2}{\sigma^2(M+2)}} < (M+2)^{1-\nu}$, which is guaranteed by \eqref{thmmusiceq1}.
\end{proof}

\section{Proof of propositions}
\label{sec:props}

\subsection{Proof of Proposition \ref{prop:upper}}

\label{proof:upper}

\begin{proof}
	The argument relies on the variational form for the minimum singular value,
	\[
	\sigma_{\min}(\Phi)
	=\min_{u\in\C^S, u\not=0} \frac{\|\Phi u\|_2}{\|u\|_2}.
	\]
	To obtain an upper bound, it suffices to consider a specific $u$, and our choice is inspired by Donoho \cite{donoho1992superresolution}. Without loss of generality, we assume that $\omega=0$. We re-index the set $\Omega=\{\omega_j\}_{j=1}^S$ so that 
	\[
	\omega_j=\frac{(j-1)\alpha}{M}
	\quad\text{for}\quad 1\leq j\leq\lambda. 
	\]
	We consider the vector $u\in\C^S$ defined by the formula
	\[
	u_j
	:=
	(-1)^{j-1} {\lambda-1\choose j-1} 
	\quad\text{for}\quad 1\leq j\leq \lambda,
	\]
	and $u_j=0$ otherwise. Note that 
	\[
	\|u\|_2={{2\lambda-2}\choose {\lambda-1}}^{1/2}.
	\]
	By the variational form for the minimum singular value, we have
	\begin{equation}
	\label{eq:upper1}
	\sigma_{\min}(\Phi)
	\leq \frac{\|\Phi u\|_2}{\|u\|_2}
	={{2\lambda-2}\choose {\lambda-1}}^{-1/2}\|\Phi u\|_2. 
	\end{equation}
	To estimate $\|\Phi u\|_2$, we identify $u$ with the discrete measure 
	\[
	\mu := \sum_{j=1}^{\lambda} u_j\delta_{(j-1)\alpha/M}.
	\]
	We also define a modulated Dirichlet kernel $D_M\in C^\infty(\T)$ by the formula, $D_M(\omega):=\sum_{m=0}^M e^{2\pi im\omega}$. We readily check that
	\begin{equation}
	\label{eq:upper2}
	\|\Phi u\|_2
	=\sum_{m=0}^M |(\Phi u)_m|^2
	=\(\sum_{m=0}^M |\hat\mu(m)|^2\)^{1/2}
	=\|\mu*D_M\|_{L^2(\T)}. 
	\end{equation}
	We see that all $\omega\in\T$,
	\begin{equation}
	\label{eq:upper2.1}
	(\mu*D_M)(\omega)
	=\sum_{j=0}^{\lambda-1} (-1)^j {\lambda-1\choose j} D_M\(\omega-\frac{j\alpha}{M}\). 
	\end{equation}
	The right hand side is the $(\lambda-1)$-th order backwards finite difference of $D_M$. It is well-known that for each $\omega\in\T$, we have
	\begin{equation}
	\label{eq:upper3}
	(\mu*D_M)(\omega)
	=\(\frac{\alpha}{M}\)^{\lambda-1} D_M^{(\lambda-1)}(\omega)+R_{\lambda-1}(\omega),
	\end{equation}
	where $D_M^{(\lambda-1)}$ denotes the $(\lambda-1)$-th derivative of $D_M$ and the remainder term $R_{\lambda-1}$ in magnitude is point-wise $O((\alpha/M)^\lambda)$ as $\alpha\to 0$. In order to exactly determine how small we require $\alpha$ to be, we calculate the remainder term explicitly. By a Taylor expansion of $D_M$, for each $\omega \in\T$ and $0\leq j\leq \lambda-1$, there exists $\omega_j\in (\omega-j\alpha/M,\omega)$ such that 
	\[
	D_M(\omega-j\alpha)
	=\sum_{k=0}^{\lambda-1} D_M^{(k)}(\omega) \(\frac{\alpha}{M}\)^k \frac{(-1)^kj^k}{k!} + D_M^{(\lambda)}(\omega_j) \(\frac{\alpha}{M}\)^\lambda \frac{(-1)^\lambda j^\lambda}{\lambda!}. 
	\]
	Using this formula in equations \eqref{eq:upper2} and \eqref{eq:upper2.1}, we see that 
	\[
	R_{\lambda-1}(\omega)
	=\sum_{j=0}^{\lambda-1} (-1)^{j+\lambda} {\lambda-1 \choose j} D_M^{(\lambda)}(\omega_j) \(\frac{\alpha}{M}\)^\lambda \frac{j^\lambda}{\lambda!}. 
	\] 
	We are ready to bound equation \eqref{eq:upper3} in the $L^2(\T)$ norm. By the Bernstein inequality for trigonometric polynomials, we have
	\[
	\|D_M^{(\lambda-1)}\|_{L^2(\T)}
	\leq (2\pi M)^{\lambda-1} \|D_M\|_{L^2(\T)}
	= \sqrt{M+1}\ (2\pi M)^{\lambda-1}. 
	\]
	By the same argument, we have
	\begin{align*}
	\|R_{\lambda-1}\|_{L^2(\T)}
	&\leq \sum_{j=0}^{\lambda-1} {\lambda-1 \choose j} \(\frac{\alpha}{M}\)^\lambda \frac{j^\lambda}{\lambda!} \|D_M^{(\lambda)}\|_{L^\infty(\T)} \\
	&\leq C(\lambda) \alpha (2\pi \alpha)^{\lambda-1} \|D_M\|_{L^\infty(\T)} \\
	&\leq C(\lambda)\alpha (2\pi \alpha)^{\lambda-1}(M+1). 
	\end{align*}
	Using these upper bounds together with \eqref{eq:upper3}, we have
	\[
	\|\mu*D_M\|_{L^2(\T)}
	\leq \sqrt{M+1} \ (2\pi \alpha)^{\lambda-1} \(1+ C(\lambda) \alpha \sqrt{M+1}\). 
	\]
	This inequality and the assumed upper bound for $\alpha$ \eqref{eq:alpha}, we see that
	\[
	\|\mu*D_M\|_{L^2(\T)}
	\leq 2\sqrt{M+1}\ (2\pi \alpha)^{\lambda-1}.
	\]
	Combining this inequality with \eqref{eq:upper1} and \eqref{eq:upper2} completes the the proof. 
\end{proof}

\subsection{Proof of Proposition \ref{prop:duality}}

\label{proof:duality}

\begin{proof}
	We first prove that $\P(\Omega,M,w)$ is non-empty for any $w\in\C^S$ and $S\leq M-1$. For each $1\leq k\leq S$, we consider the Lagrange polynomials,
	\[
	L_k(\omega):= \prod_{j\not=k} \frac{e^{2\pi i \omega}-e^{2\pi i\omega_j}}{e^{2\pi i\omega_k}-e^{2\pi i\omega_j}}. 
	\]
	We have $L_k(\omega_j)=\delta_{j,k}$ by definition, and after expanding $L_k$ as a summation, we see that $L_k\in \P(S-1)\subset \P(M)$. This implies $\sum_{k=1}^S w_k L_k \in \P(\Omega,S-1,w)$, which proves the first part of the proposition. 
	
	Let $v$ be any unit norm vector such that $\|\Phi v\|_2=\sigma_{\min}(\Phi)$. The set of all trigonometric polynomials $f\in\P(M)$ can be written in the form 
	\[
	f(\omega)
	=\sum_{m=0}^{M-1} \hat f(m) e^{2\pi i m\omega}.
	\]
	Then $f\in \P(M,\Omega,v)$ if and only if $f\in\P(M)$ and it Fourier coefficients satisfy the under-determined system of equations,
	\[
	v_j
	=\sum_{m=0}^{M-1} \hat f(m) e^{2\pi im\omega_j}
	\quad\text{for}\quad
	1\leq j\leq S.
	\]
	Since $\P(\Omega,M,v)$ is non-empty, pick any $f\in\P(\Omega,M,v)$. Since $\|f\|_{L^2(\T)}=\|\hat f\|_{\ell^2(\Z)}$, the functions $f\in\P(M)$ that satisfy this system of equations and have minimal $L^2(\T)$ norm are the ones with Fourier coefficients given by the Moore-Penrose pseudo-inverse solution to the above system of equations. Namely,
	\[
	\min_{f\in \P(M,\Omega,v)} \|f\|_{L^2(\T)}
	=\min_{\Phi^*u=v} \|u\|_2
	=\|(\Phi^*)^\dagger v\|_2
	=\frac{1}{\sigma_{\min}(\Phi)}.
	\]
	Rearranging this inequality completes the proof of the proposition.
\end{proof}

\subsection{Proof of Proposition \ref{prop:duality2}}

\label{proof:duality2}

\begin{proof}
	Define the measure $\mu=\sum_{j=1}^S v_j\delta_{\omega_j}$, and note that $\hat\mu(m)=(\Phi v)_m$. We have
	\[
	\Big|\int_{\T} \overline f\ d\mu\Big|
	=\Big|\sum_{j=1}^S \overline{f(\omega_j)} v_j \Big|
	=\Big|\|v\|_2^2+ \sum_{j=1}^S \overline{v_j} \epsilon_j\Big|
	\geq \|v\|_2^2-\|v\|_2\|\epsilon\|_2
	=1-\|\epsilon\|_2. 
	\]
	On the other hand, using that $f\in\P(M)$, Cauchy-Schwartz, and Parseval,
	\[
	\Big|\int_{\T} \overline f\ d\mu\Big|
	=\Big|\sum_{m=1}^{M-1} \overline{\hat f(m)}\hat\mu(m)\Big| 
	\leq \|\hat f\|_{\ell^2(\Z)}\|\Phi v\|_2
	=\|f\|_{L^2(\T)} \|\Phi v\|_2. 
	\]
	Combining the previous two inequalities completes the proof.
\end{proof}

\subsection{Proof of Proposition \ref{prop:demanet}}

\label{proof:demanet}

For this proof, we make the following changes to the notation. We can identify every discrete measure $\mu$ whose support is contained in the grid with spacing $1/N$ and consists of $S$ points with a $S$-sparse vector $x\in\C^N$. Under this identification, the Fourier transform of $\mu$ is identical to the discrete Fourier transform of $x$. Let $\C_S^N$ be the set of $S$-sparse vectors in $\C^N$, and $\calF$ be the first $M+1$ rows of the $N\times N$ discrete Fourier transform matrix. With this notation at hand, the min-max error is
\[
\calE(M,N,S,\delta)
=\inf_{\varphi\in\calA}\ \sup_{y(x,\eta)\in\calY}  \|\varphi_y-x\|_2.
\] 

\begin{proof}
	We prove the upper bound first. Let $\varphi$ be the function that maps each $y\in\calY$ to the sparsest vector $\varphi_y\in\C^N$ such that $\|\calF \varphi_y-y\|_2\leq \delta$. If there is not a unique choice of vector $\varphi_y$, just choose any one of them arbitrarily. Note that $\varphi_y$ exists because $x$ also satisfies the constraint that $\|\calF x-y\|_2\leq \delta$, and the choice of $\varphi_y$ does not explicitly depend on $x$ and $\eta$. Note that $\|\tilde x\|_0\leq \|\varphi_y\|_0\leq S$ by definition of $\varphi$. Then we have
	\[
	\calE(M,N,S,\delta)
	\leq \sup_{y(x,\eta)\in\cal Y} \|\varphi_y-x\|_2. 
	\]
	For any $x\in\C^N_S$ and $\eta$ with $\|\eta\|_2\leq \delta$, we have $\varphi_y-x\in \C^N_{2S}$ and
	\[
	\Theta(M,N,2S)
	\leq \frac{\|\calF(\varphi_y-x)\|_2}{\|\varphi_y-x\|_2}
	\leq \frac{\|\calF \varphi_y-y\|_2+\|\calF x-y\|_2}{\|\varphi_y-x\|_2}
	\leq \frac{2\delta}{\|\varphi_y-x\|_2}. 
	\]
	Combining the previous two inequalities and rearranging completes the proof of the upper bound for the min-max error. 
	
	We focus our attention on the lower bound for the min-max error. By definition of the smallest singular value, there exists $v\in\C^N_{2S}$ of unit norm such that 
	\[
	\Theta(M,N,2S)
	={\|\calF v\|_2}.
	\]
	Pick any vectors $v_1,v_2\in \C^N_S$ such that 
	\[
	\frac{\delta}{\Theta(M,N,2S)}\ v=v_1-v_2.
	\]
	Suppose we are given the data
	\[
	y=\calF v_1=\calF v_2+\calF(v_1-v_2).
	\]
	Let $\eta := \calF(v_1-v_2)\in\C^{M+1}$. The previous three equations imply
	\[
	\|\eta\|_2
	=\|\calF(v_1-v_2)\|_2
	=\frac{\delta}{\Theta(M,N,2S)} \|\calF v\|_2
	\leq\delta.
	\]
	This proves that $y$ is both the noiseless first $M$ Fourier coefficients of $v_1$ as well as the noisy first $M$ Fourier coefficients of $v_2$ with noise $\calF(v_1-v_2)$ with noise $\eta$. Thus, we have $y\in\calY$ with $y=y(v_1,0)$ and $y=y(v_2,\eta)$. Consequently, we have
	\[
	\calE(M,N,S,\delta)
	\geq \inf_{\varphi\in\calA}\max_{k=1,2} \|f(y)-v_k\|. 
	\]
	Using that $v$ has unit norm, for any $\varphi\in\calA$, we have
	\[
	\frac{\delta}{\Theta(M,N,2S)}
	=\|v_1-v_2\|_2
	\leq \|\varphi_y-v_1\|_2+\|\varphi_y-v_2\|_2
	\leq 2 \max_{k=1,2} \|\varphi_y-v_k\|_2.
	\]
	This holds for all $f\in\calA$, so combining the previous two inequalities completes the proof of the lower bound for the min-max error. 
\end{proof}

\section{Proof of lemmas}
\label{sec:lemmas}

\subsection{Proof of Lemma \ref{lemma:localizedpoly}}

\label{proof:localizedpoly}

\begin{proof}
	Fix a $\omega_j\in\Omega$, and so $\omega_j\in\Lambda_a$ for some $1\leq a\leq A$. We explicitly construct each $I_j$, and it is more convenient to break the construction into two cases. 
	
	The simpler case is when $\lambda_a=1$. Note that $B_a=\rho_j=1$. Then we simply set 
	\[
	I_j(\omega) := e^{2\pi i M(\omega-\omega_j)} F_M(\omega-\omega_j),
	\]
	where we recall that $F_M$ is the Fej\'er kernel. We trivially have $I_j(\omega_k)=\delta_{j,k}$ for all $\omega_k\in \Lambda_a$ and $I_j\in\P(M)$. Using the point-wise bound for the Fej\'er kernel (\ref{eq:fejer2}) and the cluster separation condition (\ref{eq:sep1}), we have
	\[
	|I_j(\omega_k)|
	\leq \frac{1}{4(M+1)^2|\omega_k-\omega_j|_\T^2}
	\leq \frac{1}{400S}.
	\]
	Using the $L^2$ norm bound for the Fej\'er kernel (\ref{eq:fejer1}), we see that
	\[
	\|I_j\|_{L^2(\T)}
	\leq \frac{1}{\sqrt{M+1}}.
	\]
	This completes the proof of the lemma when $\lambda_a=1$. 
	
	From here onwards, we assume that $\lambda\geq 2$. To define $I_j$, we must construct two axillary functions $G_j$ and $H_j$. We define the Lagrange-like polynomial,
	\[
	G_j(\omega)
	:=\prod_{\omega_k\in \Lambda_a\setminus\{\omega_j\}}\frac{e^{2\pi i Q_jt}-e^{2\pi i Q_j\omega_k}}{e^{2\pi i Q_j\omega_j}-e^{2\pi i Q_j\omega_k}},
	\quad\text{where}\quad
	Q_j:=\Big\lfloor \frac{M}{\lambda_a}\Big\rfloor.
	\]
	Note that $Q_j$ is positive because $M/\lambda_a\geq M/S\geq 1$. This function is well-defined because its denominator is always non-zero: this follows from the observation that the inequalities, $Q_j\leq M/2$ and $|\omega_j-\omega_k|_\T<1/M$, imply
	\[
	|Q_j\omega_j-Q_j\omega_k|_\T
	=Q_j|\omega_j-\omega_k|_\T.
	\]
	By construction, the function $G_j$ satisfies the important property that 
	\begin{equation}
	\label{eq:Gdelta}
	G_j(\omega_k)=\delta_{j,k},
	\quad\text{for all}\quad 
	\omega_k\in\Lambda_a.
	\end{equation}
	
	We upper bound $G_j$ in the sup-norm. We begin with the estimate 
	\[
	\|G_j\|_{L^\infty(\T)}
	\leq \prod_{\omega_k\in \Lambda_a\setminus\{\omega_j\}} \frac{2}{|1-e^{2\pi i Q_j(\omega_j-\omega_k)}|}.
	\]
	Recall the trigonometric inequality,
	\[
	2-2\cos(2\pi t)\geq (2\pi t)^2\(1-\frac{\pi^2 t^2}{3}\)
	\quad\text{for } t\in[-1/2,1/2],
	\]
	which follows from a Taylor expansion of cosine. Using this inequality, we deduce the bound,
	\[
	\|G_j\|_{L^\infty(\T)}
	\leq \prod_{\omega_k\in \Lambda_a\setminus\{\omega_j\}} \frac{1}{\pi Q_j|\omega_j-\omega_k|_\T} \(1-\frac{\pi^2 Q_j^2|\omega_j-\omega_k|_\T^2}{3}\)^{-1/2}.
	\]
	Since $Q_j= \lfloor M/\lambda_a \rfloor$ and $|\omega_j-\omega_k|_\T< 1/M$, we have 
	\begin{gather}
	\label{eq:Gpoint}
	\begin{split}
	\|G_j\|_{L^\infty(\T)}
	&\leq \(1-\frac{\pi^2}{3\lambda_a^2}\)^{-(\lambda_a-1)/2} \prod_{\omega_k\in \Lambda_a\setminus\{\omega_j\}} \frac{1}{\pi Q_j|\omega_j-\omega_k|_\T} 
	= \tilde B_a \lambda_a^{\lambda_a-1}\rho_j.	
	\end{split}
	\end{gather}
	
	We next define the function $H_j$ by the formula,
	\[
	H_j(\omega)
	:=\big(e^{2\pi iP_j(\omega-\omega_j)} F_{P_j}(\omega-\omega_j)\big)^{\lambda_a},
	\quad\text{where}\quad
	P_j := \Big\lfloor \frac{M}{2\lambda_a^2} \Big\rfloor. 
	\]
	Recall that $F_{P_j}$ denotes the Fej\'er kernel and note that $P_j$ is positive because $M/(2\lambda_a^2)\geq M/(2S^2)\geq 1$. We need both a decay and norm bound for $H_j$. To obtain a norm bound, we use H\"older's inequality, that the Fej\'er kernel is point-wise upper bounded by 1, the norm bound for the Fej\'er kernel \eqref{eq:fejer1}, and the inequality $P_j+1\geq M/(2\lambda_a^2)$, to obtain, 
	\begin{equation}
	\label{eq:Hnorm}
	\|H_j\|_{L^2(\T)}
	\leq \|F_{P_j}\|_{L^\infty(\T)}^{\lambda_a-1} \|F_{P_j}\|_{L^2(\T)}
	\leq \frac{1}{\sqrt{P_j+1}}
	\leq \(\frac{2\lambda_a^2}{M}\)^{1/2}.	
	\end{equation}
	To obtain a decay bound for $H_j$, we use the point-wise bound for the Fej\'er kernel \eqref{eq:fejer2} to deduce, 
	\[
	|H_j(\omega)|
	\leq \(\frac{1}{2(P_j+1)|\omega-\omega_j|_\T}\)^{2\lambda_a}
	\leq \(\frac{\lambda_a^2}{M|\omega-\omega_j|_\T}\)^{2\lambda_a},
	\quad\text{for all}\quad
	\omega\in\T.
	\]
	We would like to specialize this to the case that $\omega=\omega_j$ for $\omega_j\not\in\Lambda_a$. We need to make the following observations first. Observe that $1\leq \lfloor t \rfloor / t\leq 2$ for any $t\geq 1$. Using this inequality and that $\lambda_a\geq 2$, we see that
	\[
	(20 \tilde B_a)^{1/(2\lambda_a)}
	\leq 20^{1/(2\lambda_a)} \(1-\frac{\pi^2}{3\lambda_a^2}\)^{-1/4+1/(4\lambda_a)} 2^{(\lambda_a-1)/(2\lambda_a)}
	\leq 10. 
	\]
	This inequality and the cluster separation condition \eqref{eq:sep1} imply
	\[
	|\omega_k-\omega_j|_\T
	\geq \frac{10\lambda_a^2 (S\lambda_a^{\lambda_a-1} \rho_j)^{1/(2\lambda_a)}}{M}
	\geq \frac{\lambda_a^2(20 \tilde B_a S\lambda_a^{\lambda_a-1}\rho_j)^{1/(2\lambda_a)}}{M}
	\quad\text{for all}\quad
	\omega_k\not\in\Lambda_a. 
	\]
	Combining this with the previous upper bound on $H_j$ shows that 
	\begin{equation}
	\label{eq:Homega_k}
	|H_j(\omega_k)|
	\leq \frac{1}{20S \tilde B_a \lambda_a^{\lambda_a-1}\rho_j}
	\quad\text{for all}\quad
	\omega_k\not\in\Lambda_a.
	\end{equation}	
	
	We define the function $I_j$ by the formula 
	\[
	I_j(\omega)
	:=G_j(\omega)H_j(\omega). 
	\]
	It follows immediately from the property \eqref{eq:Gdelta} that 
	\[
	I_j(\omega_k)=\delta_{j,k}
	\quad\text{for all}\quad 
	\omega_k\in\Lambda_a.
	\]
	The negative frequencies of $I_j$ are zero, while its largest non-negative frequency is bounded above by
	\begin{align*}
	2P_j\lambda_a + (\lambda_a-1) Q_j
	&\leq \frac{M}{\lambda_a}+ (\lambda_a-1)\(\frac{M}{\lambda_a}\)
	\leq M,
	\end{align*}
	which proves that $I_j\in\P(M)$. We use H\"older's inequality, the sup-norm bound for $G_j$ \eqref{eq:Gpoint}, and the norm bound for $H_j$ \eqref{eq:Hnorm} to see that 
	\[
	\|I_j\|_{L^2(\T)}
	\leq \|G_j\|_{L^\infty(\T)}\|H_j\|_{L^2(\T)}
	\leq \tilde B_a \lambda_a^{\lambda_a} \rho_j \(\frac{2}{M}\)^{1/2}. 
	\]
	Finally, we use the sup-norm bound for $G_j$ \eqref{eq:Gpoint} and the bound for $|H_j(\omega_k)|$ \eqref{eq:Homega_k} to see that
	\[
	|I_j(\omega_k)|
	\leq \|G_j\|_{L^\infty(\T)}|H_j(\omega_k)|
	\leq \frac{1}{20S}
	\quad\text{for all}\quad
	\omega_k\not\in\Lambda_a. 
	\]
	
\end{proof}

\subsection{Proof of Lemma \ref{lemma:poly1}}

\label{proof:poly1}

\begin{proof}
	Fix integers $M,N,S$ satisfying the assumptions of Lemma \ref{lemma:poly1}. Fix a support set $\Omega$, contained in the grid with spacing $1/N$ and of cardinality $S$. We do a two-scale analysis. For each  $\omega_j\in\Omega$, we define the discrete sets and integers,
	\begin{align*}
	\Gamma_j &:= \Gamma_j(\Omega) = \Big\{\omega_k\in\Omega\colon |\omega_k-\omega_j|_\T<\frac{1}{M} \Big\}
	\quad\text{and}\quad
	\gamma_j:=|\Gamma_j|, \\
	\calT_j 
	&:=\calT_j(\Omega)
	=\Big\{\omega_k\in\Omega\colon |\omega_k-\omega_j|_\T<\frac{S}{2M} \Big\}
	\quad\text{and}\quad
	\tau_j:=|\calT_j|.
	\end{align*}
	
	To construct $H_j(\cdot,\Omega)$, we need to define two axillary functions, similar to the construction done in Lemma \ref{lemma:localizedpoly}. We define the integers
	\[
	Q_{j,k}
	:=Q_{j,k}(\Omega)
	:=\begin{cases}
	\ \lfloor M/S \rfloor &\text{if } \omega_k\in\calT_j\setminus\{\omega_j\}, \medskip \\
	\  \lfloor 1/(2|\omega_j-\omega_k|_\T) \rfloor  &\text{if } \omega_k\in\Omega\setminus \calT_j. 
	\end{cases}
	\]
	We readily verify that we have the inequalities $1\leq Q_{j,k}\leq M/S$ and
	\begin{equation}
	\label{eq:Q}
	|Q_{j,k}\omega_j-Q_{j,k}\omega_k|_\T
	= Q_{j,k}|\omega_j-\omega_k|_\T
	\quad\text{for all}\quad
	\omega_j,\omega_k\in\Omega. 
	\end{equation}
	This observation implies that the Lagrange-like polynomial, 
	\[
	G_j(\omega)
	:=G_j(\omega,\Omega)
	:=\prod_{\omega_k\in \Omega\setminus\{\omega_j\}} \frac{e^{2\pi i Q_{j,k}\omega}-e^{2\pi i Q_{j,k}\omega_k}}{e^{2\pi i Q_{j,k}\omega_j}-e^{2\pi i Q_{j,k}\omega_k}}, 
	\]
	has non-zero denominators, and is thus well-defined. By construction, we have the interpolation identity,
	\[
	G_j(\omega_k)=\delta_{j,k}
	\quad\text{for all}\quad
	\omega_j,\omega_k\in\Omega.
	\]	
	
	\item 
	We bound $G_j$ in the sup-norm. We begin with the inequality, 
	\begin{equation}
	\label{eq:Gsup1}
	\|G_j\|_{L^\infty(\T)}
	\leq \prod_{\omega_k\in \Omega\setminus\{\omega_j\}} \frac{2}{|1-e^{2\pi iQ_{j,k}(\omega_j-\omega_k)}|}.
	\end{equation}
	Recall that we have the partition,
	\[
	\Omega\setminus\{\omega_j\}
	=(\Gamma_j\setminus \{\omega_j\})\cup (\calT_j\setminus\Gamma_j)\cup (\Omega\setminus\calT_j).
	\]
	Then we break \eqref{eq:Gsup1} into three products according to this partition, and estimate each term at a time. 
	\begin{enumerate}[(a)]
		\item
		We first consider the product over $\omega_k\in\Gamma_j\setminus\{\omega_j\}$. If $\Gamma_j\setminus \{\omega_j\}=\emptyset$, there is nothing to do. Hence, assume that $\gamma_j\geq 2$. By a Taylor expansion for cosine, we obtain the inequality,
		\[
		2-2\cos(2\pi t)\geq (2\pi t)^2\(1- \frac{\pi^2 t^2}{3}\)
		\quad\text{for } t\in[-1/2,1/2].
		\]
		Using this lower bound, the observation that   $Q_{j,k}=\lfloor M/S \rfloor \leq M/S\leq M/\gamma_j$ when $\omega_k\in\Gamma_j\setminus\{\omega_j\}$, and the assumption that $|\omega_j-\omega_k|<1/M$ for all $\omega_k\in\Gamma_j$, we obtain 
		\begin{align*}
		&\prod_{\omega_k\in \Gamma_j\setminus\{\omega_j\}} \frac{2}{|1-e^{2\pi iQ_{j,k}(\omega_j-\omega_k)}|} \\
		&\quad\quad\leq \prod_{\omega_k\in \Gamma_j\setminus\{\omega_j\}} \(1-\frac{\pi^2Q_{j,k}^2|\omega_j-\omega_k|_\T^2}{3}\)^{-1/2} \prod_{\omega_k\in \Gamma_j\setminus\{\omega_j\}} \frac{1}{\pi Q_{j,k}|\omega_j-\omega_k|_\T} \\
		&\quad\quad\leq  \(1-\frac{\pi^2}{3\gamma_j^2}\)^{-(\gamma_j-1)/2} \Big\lfloor \frac{M}{S}\Big\rfloor^{-(\gamma_j-1)} \prod_{\omega_k\in \Gamma_j\setminus\{\omega_j\}} \frac{1}{\pi |\omega_j-\omega_k|_\T} \\
		&\quad\quad\leq  \(\frac{12}{12-\pi^2}\)^{1/2}\Big\lfloor \frac{M}{S}\Big\rfloor^{-(\gamma_j-1)} \prod_{\omega_k\in \Gamma_j\setminus\{\omega_j\}} \frac{1}{\pi |\omega_j-\omega_k|_\T} . 
		\end{align*}
		For the last inequality, we made the observation that $(1-\pi^2/(3t^2))^{-(t-1)/2}$ is a decreasing function of $t$ on the domain $t\geq 2$. 
		
		\item 
		We consider the product over $\omega_k\in\calT_j\setminus\Gamma_j$, and note that $Q_{j,k}=\lfloor M/S \rfloor$ for this case. Recall the trigonometric inequality 
		\begin{equation}
		\label{eq:trig}
		|e^{2\pi it}-1|\geq 4|t|_\T, \quad\text{for all}\quad t\in\R.
		\end{equation}
		We this trigonometric inequality and \eqref{eq:Q} to see that
		\begin{align*}
		\prod_{\omega_k\in\calT_j\setminus\Gamma_j} \frac{2}{|1-e^{2\pi iQ_{j,k}(\omega_j-\omega_k)}|}
		&\leq \prod_{\omega_k\in\calT_j\setminus\Gamma_j} \frac{1}{2Q_{j,k}|\omega_j-\omega_k|_\T} \\
		&\leq \Big\lfloor \frac{M}{S}\Big\rfloor^{-\tau_j+\gamma_j}\(\frac{1}{2}\)^{\tau_j-\gamma_j} \prod_{\omega_k\in\calT_j\setminus\Gamma_j} \frac{1}{|\omega_j-\omega_k|_\T}. 
		\end{align*}
		
		\item 
		For the product over $\omega_k\in\Omega\setminus \calT_j$, note that $Q_{j,k}|\omega_j-\omega_k|_\T\geq 1/4$. Using this and the  trigonometric inequality \eqref{eq:trig} again, we see that
		\begin{align*}
		\prod_{\omega_k\in \Omega\setminus \calT_j}
		\frac{2}{|1-e^{2\pi iQ_{j,k}(\omega_j-\omega_k)}|}
		&\leq \prod_{\omega_k\in \Omega\setminus \calT_j} \frac{1}{2Q_{j,k}|\omega_j-\omega_k|_\T}
		\leq 2^{S-\tau_j}. 
		\end{align*}
	\end{enumerate}
	Combining the above three inequities with inequality \eqref{eq:Gsup1} and simplifying, we obtain an upper bound
	\begin{equation}
	\label{eq:Gsup2}
	\|G_j\|_{L^\infty(\T)} 
	\leq \(\frac{12}{12-\pi^2}\)^{1/2} \Big\lfloor \frac{M}{S}\Big\rfloor^{-\tau_j+1} \(\frac{1}{\pi}\)^{\gamma_j-1} 2^{S-2\tau_j+\gamma_j} \prod_{\omega_k\in\calT_j\setminus\{\omega_j\}} \frac{1}{|\omega_j-\omega_k|_\T}. 
	\end{equation}	
	
	\item 
	Let $P=\lfloor M/(2S) \rfloor$ and note that $P\geq 1$ because $M\geq 2S$. Let $F_P$ be the Fej\'er kernel, and by the $L^2(\T)$ bound for the Fej\'er kernel and the observation that $P+1\geq M/(2S)$, we have
	\begin{equation}
	\label{eq:Fnorm}
	\|F_P\|_{L^2(\T)}
	\leq \(\frac{1}{P+1}\)^{1/2}
	\leq \(\frac{2S}{M} \)^{1/2}. 
	\end{equation}

	Finally, we define $H_j$ by the formula, 
	\[
	H_j(\omega)
	:=H_j(\omega,\Omega)
	:=e^{2\pi iP(\omega-\omega_j)} F_P(\omega-\omega_j)G_j(\omega). 
	\]
	We still have the interpolation property that
	\[
	H_j(\omega_k)=\delta_{j,k}
	\quad\text{for all}\quad
	\omega_j,\omega_k\in\Omega.
	\]
	By construction, the negative frequencies of $H_j$ are zero while its largest positive frequency is bounded above by
	\begin{align*}
	2P + \sum_{k\not=j} Q_{j,k}
	\leq  \frac{M}{S} + \sum_{k\not=j} \frac{M}{S}
	= \frac{M}{S} +\frac{M(S-1)}{S}
	\leq M.
	\end{align*}
	This proves that $H_j\in\P(M)$. 
	
	It remains to upper bound $\sum_{j=1}^S \|H_j\|_{L^2(\T)}^2$. By H\"older's inequality and the inequalities,  
	\[
	\(\sum_{j=1}^S \|H_j\|_{L^2(\T)}^2\)^{1/2}
	\leq \(\sum_{j=1}^S \|F_P\|_{L^2(\T)}^2\|G_j\|_{L^\infty(\T)}^2\)^{1/2}
	\leq \(\frac{24}{12-\pi^2}\)^{1/2}\(\frac{S}{M} \)^{1/2} E(\Omega)^{1/2}, 
	\]
	where the constant $E(\Omega)$ is defined as
	\begin{equation}
	\label{eq:EOmega}
	E(\Omega)
	:=\sum_{j=1}^S \Big\lfloor \frac{M}{S}\Big\rfloor^{-2\tau_j+2} \(\frac{1}{\pi^2}\)^{\gamma_j-1}  4^{S-2\tau_j+\gamma_j} \prod_{\omega_k\in\calT_j\setminus\{\omega_j\}} \frac{1}{|\omega_j-\omega_k|_\T^2}.
	\end{equation}
	To complete the proof of the lemma, we need to obtain the appropriate bound on $E(\Omega)$ uniformly in $\Omega$. This is handled separately in Lemma \ref{lemma:E}, which is stated below and proved in Appendix \ref{proof:E}. 
\end{proof}

\begin{lemma}
	\label{lemma:E}
	Suppose the assumptions of Lemma \ref{lemma:poly1} hold and let $E(\Omega)$ be the quantity defined in \eqref{eq:EOmega}. Then
	\[
	E(\Omega)
	\leq \Big\lfloor \frac{M}{S}\Big\rfloor^{-2S+2} N^{2S-2} \(\frac{1}{\pi}\)^{2S-2}   \sum_{j=1}^S  \prod_{k\not=j} \frac{1}{(j-k)^2}.
	\]
\end{lemma}

\subsection{Proof of Lemma \ref{lemma:E}}

\label{proof:E}

Before we prove the lemma, we motivate the argument that we are about to use. We view $E(\Omega)$ as a function defined on all $N\choose S$ possible sets $\Omega$ supported on the grid with spacing $1/N$ and of cardinality $S$. To upper bound $E(\Omega)$ uniformly in $\Omega$, one method is to determine which $\Omega$ attain(s) the maximum. The maximizer is clearly not unique, since $E(\Omega)$ is invariant under cyclic shifts of $\Omega$ by $1/N$. However, we shall argue that the maximizer is attained by shifts of $\Omega_*=\{n/N\}_{n=0}^{S-1}$. Note that
\begin{equation}
\label{eq:E0}
E(\Omega_*)
=\Big\lfloor \frac{M}{S}\Big\rfloor^{-2S+2} N^{2S-2} \(\frac{1}{\pi}\)^{2S-2}   \sum_{j=1}^S  \prod_{k\not=j} \frac{1}{(j-k)^2}.
\end{equation}

Thus, the lemma is complete once we prove that $E(\Omega)\leq E(\Omega_*)$. While it seems intuitive that $E(\Omega)\leq E(\Omega_*)$ for all $\Omega$, it is not straightforward to prove. When $\Omega$ is contained in a small interval, the product over $\calT_j$ in the definition of $E(\Omega)$ given in \eqref{eq:EOmega} is large, but that is offset by the remaining terms, which are small. The major difficulty is that $E(\Omega)$ is highly dependent on the configuration of $\Omega$. If we perturb just one of the $\omega_j\in\Omega$ and keep the rest fixed, it is possible for all $S$ terms in the summation in the definition of $E(\Omega)$ to change. This makes continuity and perturbation arguments difficult to carry out. To deal with this difficultly, we proceed with the following extension argument. 

\begin{proof}
	
	We extend $E$ to a function of $D=S(S-1)$ variables in the following way. We write $w\in\R^D$ to denote the $D$ variables $\{w_{j,k}\}_{1\leq j,k\leq S,j\not=k}$. We do not impose that $\{w_{j,k}\}_{j\not=k}$ are unique, that $w_{j,k}=w_{k,j}$, or that they lie on some grid. They are just $D$ independent real variables for now. We define the sets and integers,
	\begin{align*}
	A_j(w)
	&:=\Big\{w_{j,k}\colon w_{j,k}<\frac{S}{2M}\Big\} \quad\text{and}\quad  a_j(w):=|A_j(w)|,\\
	B_j(w)
	&:=\Big\{w_{j,k}\colon w_{j,k}<\frac{1}{M}\Big\} \quad\text{and}\quad b_j(w):=|B_j(w)|.
	\end{align*}
	We define the function $F\colon\R^D\to\R$ by the formula,
	\begin{equation}
	\label{eq:F}
	F(w)
	:=\sum_{j=1}^S \Big\lfloor \frac{M}{S}\Big\rfloor^{-2a_j(w)+2} \(\frac{1}{\pi^2}\)^{b_j(w)-1} 4^{S-2a_j(w)+b_j(w)} \prod_{w_{j,k}\in A_j(w)} \frac{1}{w_{j,k}^2}. 
	\end{equation}
	We restrict $F$ to the domain $[1/N,1/2]^D \cap H$, where
	\[
	H
	:=\bigcap_{k=1}^S \Big\{w\in\R^D\colon \sum_{j\not=k} w_{j,k} \geq \frac{c(S)}{N} \Big\},
	\]
	and the constant $c(S)$ is defined as
	\[ 
	c(S) 
	:= \begin{cases}
	\ 2\(1+2+\dots+\frac{S-1}{2}\)\ &\text{if } $S$ \text{ is odd}, \\
	\ 2\(1+2+\dots+\frac{S-2}{2}\)+\frac{S}{2} &\text{if } $S$ \text{ is even}.
	\end{cases} 
	\]
	
	We argue that $F$ is an extension of $E$. Note that any $\Omega$ can be mapped to a $w(\Omega)\in \R^D$ via the relationship $(w(\Omega))_{j,k}=|\omega_j-\omega_k|_\T$ for all $j\not=k$. Under this mapping, we have $a_j(w)=\tau_j$ and $b_j(w)=\gamma_j$, which shows that 
	\[
	F(w(\Omega))=E(\Omega).
	\]
	Moreover, $w(\Omega)$ is clearly contained in $[1/N,1/2]^D$. For each $1\leq k\leq S$, we have
	\[
	\sum_{j\not=k} (w(\Omega))_{j,k} 
	=\sum_{j\not=k} |\omega_j-\omega_k|_\T
	\geq \frac{c(S)}{N}. 
	\]
	This inequality implies that $w(\Omega)$ is contained in the set $[1/N,1/2]^D\cap H$. Thus, $F$ is indeed an extension of $E$, and for all $\Omega$, we have
	\begin{equation}
	\label{eq:E1}
	E(\Omega)
	=F(w(\Omega))
	\leq \sup_{w\in [1/N,1/2]^D\cap H} F(w). 
	\end{equation}
	
	We remark that there is a clear advantage of working with $F$ instead of $E$. If one coordinate of $w$ is perturbed while the rest of the $D-1$ coordinates of $w$ remain fixed, then only one of the $S$ terms in the summation in \eqref{eq:F} is perturbed. 
	
	Observe that $[1/N,1/2]^D\cap H$ is compact because it is the intersection of a closed cube with $S$ closed half-spaces. Clearly $F$ is continuous on the domain $[1/N,1/2]^D\cap H$, so the supremum of $F$ is attained at some point in this set. We first simplify matters and prove that
	\begin{equation}
	\label{eq:F1}
	\max_{w\in [1/N,1/2]^D\cap H} F(w)
	=\max_{w\in [1/N,1/M]^D\cap H} F(w),
	\end{equation}
	which is done via the following two reductions.
	
	\begin{enumerate}[(a)]
		\item 
		Our first claim is that
		\[
		\max_{w\in [1/N,1/2]^D\cap H} F(w)
		=\max_{w\in [1/N,S/(2M)]^D\cap H} F(w).
		\]
		Suppose for the purpose of yielding a contradiction, the maximum of $F$ is not attained at any point in $[1/N,S/(2M)]^D\cap H$. This is equivalent to the claim that, for any maximizer $w$ of $F$, there exist indices $(m,n)$ such that $a_m(w)\leq S-1$ and $w_{m,n}>S/(2M)$. We define the vector $v\in[1/N,1/2]^D\cap H$ by the relationship
		\[
		v_{j,k}:=
		\begin{cases}
		\ S/N &\text{if } (j,k)=(m,n), \\
		\ w_{j,k} &\text{otherwise}.
		\end{cases}
		\]
		Since $v$ and $w$ agree except at one coordinate, we readily calculate that 
		\begin{align*}
		&F(w)-F(v) \\
		&=
		\Big\lfloor \frac{M}{S}\Big\rfloor^{-2a_m(w)+2} \(\frac{1}{\pi^2}\)^{b_m(w)-1} 4^{S-2a_m(w)+b_m(w)} \(\prod_{w_{j,k}\in A_j(w)} \frac{1}{w_{m,k}^2}\)\(1-\Big\lfloor \frac{M}{S}\Big\rfloor^{-2} \frac{1}{4\pi^2 v_{m,n}^2}\).
		\end{align*}
		The assumption that $N\geq \pi MS$ and $S\geq 2$ imply
		\[
		\frac{1}{4\pi^2 v_{m,n}^2}
		=\frac{N^2}{4\pi^2S^2}
		\geq \(\frac{M}{2}\)^2
		\geq \Big\lfloor \frac{M}{S}\Big\rfloor^2.
		\]
		This proves that $F(w)\leq F(v)$, which is a contradiction.  
		
		\item 
		Our second claim is that 
		\[
		\max_{w\in [1/N,S/(2M)]^D\cap H} F(w)
		=\max_{w\in [1/N,1/M]^D\cap H} F(w).
		\]
		Suppose for the purpose of yielding a contradiction, the maximum of $F$ is not attained at any point in $[1/N,1/M]^D\cap H$. This is equivalent to  to the claim that, for any maximizer $w\in[1/N,S/(2M)]^D\cap H$, there exist indices $(m,n)$ such that $b_m(w)\leq S-1$ and $w_{m,n}\geq 1/M$. We define the vector $v\in [1/N,1/M]^D\cap H$ by the relationship,
		\[
		v_{j,k}=
		\begin{cases}
		\ S/N &\text{if } (j,k)=(m,n)\\
		\ w_{j,k} &\text{otherwise}. 
		\end{cases}
		\]
		Since $v$ and $w$ agree except at one coordinate, we see that
		\begin{align*}
		F(w)-F(v)
		=\Big\lfloor \frac{M}{S}\Big\rfloor^{-2S+2} \(\frac{1}{\pi^2}\)^{b_m(w)-1}  4^{-S+b_m(w)} \(\prod_{k\in A_m(w)\setminus\{n\}} \frac{1}{w_{m,k}^2}\) \(\frac{1}{w_{m,n}^2}-\frac{4}{\pi^2 v_{m,n}^2}\).
		\end{align*}
		The assumption that $N\geq \pi MS$ implies
		\[
		\frac{4}{\pi^2 v_{m,n}^2}
		=\frac{4N^2}{\pi^2 S^2}
		\geq 4M^2
		\geq \frac{1}{w_{m,n}^2}. 
		\]
		This shows that $F(w)\leq F(v)$, which is a contradiction.		
	\end{enumerate}
	
	Thus, we have established \eqref{eq:F1}, and combining this fact with \eqref{eq:E1} yields,
	\begin{equation}
	\label{eq:E2}
	E(\Omega)
	=F(w(\Omega))
	\leq \max_{w\in [1/N,1/M]^D\cap H} F(w).
	\end{equation}
	When $w\in [1/N,1/M]^D\cap H$, the function $F$ reduces to
	\[
	F(w)
	= \Big\lfloor \frac{M}{S}\Big\rfloor^{-2S+2} \(\frac{1}{\pi^2}\)^{S-1} \sum_{j=1}^S \prod_{k\in B_j(w)} \frac{1}{w_{j,k}^2}.
	\]
	Since $F$ is a smooth function of $w$, a straightforward calculation shows that each partial derivative of $F$, with respect to the canonical basis of $\R^D$, is strictly negative on $[1/N,1/M]^D\cap H$. Thus, the maximum of $F$ is attained on the boundary of $[1/N,1/M]^D\cap H$. In fact, $H$ is the intersection of $S$ half-spaces and the boundary of the $k$-th half-space is the hyperplane
	\[
	H_k
	:=\Big\{w\in\R^D\colon \sum_{j\not=k} w_{j,k} = \frac{c(S)}{N}\Big\}. 
	\]
	Since each partial derivative of $F$ is strictly negative on $[1/N,1/M]^D\cap H$, we see that the maximum of $F$ must be attained on one of these hyperplanes. We observe that $w(\Omega)$ lies on a $H_k$ if and only if $\Omega$ consists of $S$ consecutive indices. This proves that for all $\Omega$, we have
	\[
	E(\Omega)
	= F(w(\Omega))
	\leq F(w(\Omega_*))
	= E(\Omega_*).
	\]
	This combined with the formula for $E(\Omega_*)$ given in \eqref{eq:E0} completes the proof of the lemma. 
\end{proof}

\subsection{Upper bounds for discrete quantities}
\label{proof:discrete}

\begin{lemma}
	\label{lemma:comb}
	For any integer $n\geq 2$, we have
	\[
	\( \sum_{j=1}^{n} \prod_{k=1,\ k\not=j}^{n} \frac{1}{(j-k)^2}\)^{1/2}
	\leq 2\pi e \sqrt{n} \, \(\frac{n}{2}-1 \)^{-n} e^{n}. 
	\]
\end{lemma}

\begin{proof}
	Notice that 
	\[
	\min_{1\leq j\leq n} \prod_{k=1,\ k\not=j}^{n} |j-k|
	\geq 
	\begin{cases}
	\ (\frac{n}{2}-1)! (\frac{n}{2})! &\text{if $n$ is even} \\
	\ (\frac{n-1}{2})! (\frac{n-1}{2})! &\text{if $n$ is odd}. 
	\end{cases}
	\]
	By further using the well-known inequality, $k!\geq \sqrt{2\pi} \, k^{k+\frac{1}{2}}e^{-k}$ for any integer $k\geq 1$, 
	\[
	\min_{1\leq j\leq n} \prod_{k=1,\ k\not=j}^{n} |j-k|
	\geq 
	\begin{cases}
	\ 2\pi (\frac{n}{2}-1)^{\frac{n}{2}-\frac{1}{2}}(\frac{n}{2})^{\frac{n}{2}+\frac{1}{2}} e^{-n+1} &\text{if $n$ is even} \\
	\ 2\pi (\frac{n-1}{2})^{n}  e^{-n+1} &\text{if $n$ is odd}. 
	\end{cases}
	\]
	Hence, for any $n\geq 2$, we have
	\[
	\min_{1\leq j\leq n} \prod_{k=1,\ k\not=j}^{n} |j-k|
	\geq 2\pi e \(\frac{n}{2}-1 \)^{n} e^{-n}.
	\]
	Thus,
	\[
	\( \sum_{j=1}^{n} \prod_{k=1,\ k\not=j}^{n} \frac{1}{(j-k)^2}\)^{1/2}
	\leq 2\pi e \( \sum_{j=1}^{n} \(\frac{n}{2}-1 \)^{-2n} e^{2n} \)^{1/2}
	= 2\pi e \sqrt{n} \, \(\frac{n}{2}-1 \)^{-n} e^{n}. 
	\]
\end{proof}
	
\end{appendices}

\bibliographystyle{unsrt}
\bibliography{SRlimitFourierbib}

\begin{thebibliography}{10}

\bibitem{fannjiang2010remote}
Albert~C. Fannjiang, Thomas Strohmer, and Pengchong Yan.
\newblock Compressed remote sensing of sparse objects.
\newblock {\em SIAM Journal on Imaging Sciences}, 3(3):595--618, 2010.

\bibitem{fannjiang2015compressive}
Albert~C. Fannjiang.
\newblock Compressive sensing theory for optical systems described by a
  continuous model.
\newblock {\em arXiv preprint arXiv:1507.00794}, 2015.

\bibitem{fannjiang2010compressive}
Albert~C. Fannjiang.
\newblock Compressive inverse scattering: I. high-frequency simo/miso and mimo
  measurements.
\newblock {\em Inverse Problems}, 26(3):035008, 2010.

\bibitem{krim1996array}
Hamid Krim and Viberg Mats.
\newblock Two decades of array signal processing research: the parametric
  approach.
\newblock {\em IEEE Signal Processing Magazine}, 13(4):67--94, 1996.

\bibitem{schmidt1986multiple}
Ralph Schmidt.
\newblock Multiple emitter location and signal parameter estimation.
\newblock {\em IEEE Transactions on Antennas and Propagation}, 34(3):276--280,
  1986.

\bibitem{stoica1997introduction}
Petre Stoica and Randolph~L. Moses.
\newblock {\em Introduction to Spectral Analysis}, volume~1.
\newblock Prentice Hall, 1997.

\bibitem{den1997resolution}
Arnold~Jan Den~Dekker and A.~Van~den Bos.
\newblock Resolution: {A} survey.
\newblock {\em Journal of Optical Society of America}, 14(3):547--557, 1997.

\bibitem{prony1795essai}
Gaspard Riche~de Prony.
\newblock Essai exp{\'e}rimentale et analytique.
\newblock {\em Journal de L'Ecole Polytechnique}, 1(22):24--76, 1795.

\bibitem{batenkov2013accuracy}
Dmitry Batenkov and Yosef Yomdin.
\newblock On the accuracy of solving confluent prony systems.
\newblock {\em SIAM Journal on Applied Mathematics}, 73(1):134--154, 2013.

\bibitem{golub2003separable}
Gene Golub and Victor Pereyra.
\newblock Separable nonlinear least squares: the variable projection method and
  its applications.
\newblock {\em Inverse problems}, 19(2):R1, 2003.

\bibitem{beylkin2005approximation}
Gregory Beylkin and Lucas Monz{\'o}n.
\newblock On approximation of functions by exponential sums.
\newblock {\em Applied and Computational Harmonic Analysis}, 19(1):17--48,
  2005.

\bibitem{potts2010parameter}
Daniel Potts and Manfred Tasche.
\newblock Parameter estimation for exponential sums by approximate prony
  method.
\newblock {\em Signal Processing}, 90(5):1631--1642, 2010.

\bibitem{kailath1989esprit}
Richard Roy and Thomas Kailath.
\newblock {ESPRIT}-estimation of signal parameters via rotational invariance
  techniques.
\newblock {\em IEEE Transactions on Acoustics, Speech, and Signal Processing},
  37(7):984--995, 1989.

\bibitem{hua1990matrixpencil}
Yingbo Hua and Tapan~K. Sarkar.
\newblock Matrix pencil method for estimating parameters of exponentially
  damped/undamped sinusoids in noise.
\newblock {\em IEEE Transactions on Acoustics, Speech, and Signal Processing},
  38(5):814--824, 1990.

\bibitem{candes2013super}
Emmanuel~J. Cand{\`e}s and Carlos Fernandez-Granda.
\newblock Super-resolution from noisy data.
\newblock {\em Journal of Fourier Analysis and Applications}, 19(6):1229--1254,
  2013.

\bibitem{fernandez2013support}
Carlos Fernandez-Granda.
\newblock Support detection in super-resolution.
\newblock In {\em Proceedings of the 10th International Conference on Sampling
  Theory and Applications}, pages 145--148, 2013.

\bibitem{azais2015spike}
Jean-Marc Aza{\"i}s, Yohann De~Castro, and Fabrice Gamboa.
\newblock Spike detection from inaccurate samplings.
\newblock {\em Applied and Computational Harmonic Analysis}, 38(2):177--195,
  2015.

\bibitem{duval2015exact}
Vincent Duval and Gabriel Peyr{\'e}.
\newblock Exact support recovery for sparse spikes deconvolution.
\newblock {\em Foundations of Computational Mathematics}, 15(5):1315--1355,
  2015.

\bibitem{li2017elementary}
Weilin Li.
\newblock Elementary ${L}^\infty$ error estimates for super-resolution
  de-noising.
\newblock {\em arXiv preprint arXiv:1702.03021}, 2017.

\bibitem{duarte2013spectral}
Marco~F. Duarte and Richard~G. Baraniuk.
\newblock Spectral compressive sensing.
\newblock {\em Applied and Computational Harmonic Analysis}, 35(1):111--129,
  2013.

\bibitem{fannjiang2012coherence}
Albert~C. Fannjiang and Wenjing Liao.
\newblock Coherence-pattern guided compressive sensing with unresolved grids.
\newblock {\em SIAM Journal on Imaging Sciences}, 5(1):179--202, 2012.

\bibitem{bredies2013inverse}
Kristian Bredies and Hanna~Katriina Pikkarainen.
\newblock Inverse problems in spaces of measures.
\newblock {\em ESAIM: Control, Optimisation and Calculus of Variations},
  19(1):190--218, 2013.

\bibitem{boyd2017alternating}
Nicholas Boyd, Geoffrey Schiebinger, and Benjamin Recht.
\newblock The alternating descent conditional gradient method for sparse
  inverse problems.
\newblock {\em SIAM Journal on Optimization}, 27(2):616--639, 2017.

\bibitem{denoyelle2019sliding}
Quentin Denoyelle, Vincent Duval, Gabriel Peyr{\'e}, and Emmanuel Soubies.
\newblock The sliding frank--wolfe algorithm and its application to
  super-resolution microscopy.
\newblock {\em Inverse Problems}, 36(1):014001, 2019.

\bibitem{fannjiang2011music}
Albert~C. Fannjiang.
\newblock The {MUSIC} algorithm for sparse objects: a compressed sensing
  analysis.
\newblock {\em Inverse Problems}, 27(3):035013, 2011.

\bibitem{liao2016music}
Wenjing Liao and Albert Fannjiang.
\newblock {MUSIC} for single-snapshot spectral estimation: {S}tability and
  super-resolution.
\newblock {\em Applied and Computational Harmonic Analysis}, 40(1):33--67,
  2016.

\bibitem{liao2015multi}
Wenjing Liao.
\newblock Music for multidimensional spectral estimation: stability and
  super-resolution.
\newblock {\em IEEE Transactions on Signal Processing}, 63(23):6395--6406,
  2015.

\bibitem{moitra2015matrixpencil}
Ankur Moitra.
\newblock Super-resolution, extremal functions and the condition number of
  {V}andermonde matrices.
\newblock {\em Proceedings of the Forty-Seventh Annual ACM Symposium on Theory
  of Computing}, 2015.

\bibitem{fannjiang2011spie}
Albert~C. Fannjiang and Wenjing Liao.
\newblock Mismatch and resolution in compressive imaging.
\newblock {\em Wavelets and Sparsity XIV, Proceedings of SPIE}, 2011.

\bibitem{chi2011basismismatch}
Yuejie Chi, Louis~L. Scharf, Pezeshki Ali, and A.~Robert Calderbank.
\newblock Sensitivity to basis mismatch in compressed sensing.
\newblock {\em IEEE Transactions on Signal Processing}, 59(5):2182--2195, 2011.

\bibitem{candes2006robust}
Emmanuel~J. Cand{\`e}s, Justin Romberg, and Terence Tao.
\newblock Robust uncertainty principles: {E}xact signal reconstruction from
  highly incomplete frequency information.
\newblock {\em IEEE Transactions on Information Theory}, 52(2):489--509, 2006.

\bibitem{donoho2006compressed}
David~L. Donoho.
\newblock Compressed sensing.
\newblock {\em IEEE Transactions on Information Theory}, 52(4):1289--1306,
  2006.

\bibitem{morgenshtern2016super}
Veniamin~I. Morgenshtern and Emmanuel~J. Candes.
\newblock Super-resolution of positive sources: The discrete setup.
\newblock {\em SIAM Journal on Imaging Sciences}, 9(1):412--444, 2016.

\bibitem{denoyelle2017support}
Quentin Denoyelle, Vincent Duval, and Gabriel Peyr{\'e}.
\newblock Support recovery for sparse super-resolution of positive measures.
\newblock {\em Journal of Fourier Analysis and Applications}, 23(5):1153--1194,
  2017.

\bibitem{morgenshtern2020super}
Veniamin~I Morgenshtern.
\newblock Super-resolution of positive sources on an arbitrarily fine grid.
\newblock {\em arXiv preprint arXiv:2005.06756}, 2020.

\bibitem{benedetto2020super}
John~J. Benedetto and Weilin Li.
\newblock Super-resolution by means of {B}eurling minimal extrapolation.
\newblock {\em Applied and Computational Harmonic Analysis}, 48(1):218--241,
  2020.

\bibitem{donoho1992superresolution}
David~L. Donoho.
\newblock Superresolution via sparsity constraints.
\newblock {\em SIAM Journal on Mathematical Analysis}, 23(5):1309--1331, 1992.

\bibitem{demanet2015recoverability}
Laurent Demanet and Nam Nguyen.
\newblock The recoverability limit for superresolution via sparsity.
\newblock {\em arXiv preprint arXiv:1502.01385}, 2015.

\bibitem{batenkov2019super}
Dmitry Batenkov, Gil Goldman, and Yosef Yomdin.
\newblock Super-resolution of near-colliding point sources.
\newblock {\em arXiv preprint arXiv:1904.09186}, 2019.

\bibitem{odendaal1994two}
JW~Odendaal, E~Barnard, and CWI Pistorius.
\newblock Two-dimensional superresolution radar imaging using the music
  algorithm.
\newblock {\em IEEE Transactions on Antennas and Propagation},
  42(10):1386--1391, 1994.

\bibitem{stoica1995resolution}
Petre Stoica, Virginija {\v{S}}imonyte, and Torsten S{\"o}derstr{\"o}m.
\newblock On the resolution performance of spectral analysis.
\newblock {\em Signal processing}, 44(2):153--161, 1995.

\bibitem{lee1992cramer}
Harry~B Lee.
\newblock The cram{\'e}r-rao bound on frequency estimates of signals closely
  spaced in frequency.
\newblock {\em IEEE Transactions on Signal Processing}, 40(6):1507--1517, 1992.

\bibitem{lee1992eigenvalues}
Harry~B Lee.
\newblock Eigenvalues and eigenvectors of covariance matrices for signals
  closely spaced in frequency.
\newblock {\em IEEE Transactions on signal processing}, 40(10):2518--2535,
  1992.

\bibitem{turan1946rational}
Paul Tur{\'a}n.
\newblock On rational polynomials.
\newblock {\em Acta Universitatis Szegediensis, Sect. Sci. Math}, pages
  106--113, 1946.

\bibitem{alon1992uniform}
Noga Alon and Yuval Peres.
\newblock Uniform dilations.
\newblock {\em Geometric and Functional Analysis}, 2(1):1--28, 1992.

\bibitem{konyagin2000uniformly}
Sergei~V. Konyagin, Imre~Z. Ruzsa, and Wilhelm Schlag.
\newblock On uniformly distributed dilates of finite integer sequences.
\newblock {\em Journal of Number Theory}, 82(2):165--187, 2000.

\bibitem{gautschi1987lower}
Walter Gautschi and Gabriele Inglese.
\newblock Lower bounds for the condition number of {V}andermonde matrices.
\newblock {\em Numerische Mathematik}, 52(3):241--250, 1987.

\bibitem{beckermann2000condition}
Bernhard Beckermann.
\newblock The condition number of real {V}andermonde, {K}rylov and positive
  definite {H}ankel matrices.
\newblock {\em Numerische Mathematik}, 85(4):553--577, 2000.

\bibitem{eisinberg2001rectangular}
Alfredo Eisinberg, Giuseppe Franz{\'e}, and Nicola Salerno.
\newblock Rectangular {V}andermonde matrices on {C}hebyshev nodes.
\newblock {\em Linear Algebra and its Applications}, 338(1-3):27--36, 2001.

\bibitem{eisinberg2001vandermonde}
Alfredo Eisinberg, Paolo Pugliese, and Nicola Salerno.
\newblock {V}andermonde matrices on integer nodes: the rectangular case.
\newblock {\em Numerische Mathematik}, 87(4):663--674, 2001.

\bibitem{ryan2009asymptotic}
{\O}yvind Ryan and M{\'e}rouane Debbah.
\newblock Asymptotic behavior of random {V}andermonde matrices with entries on
  the unit circle.
\newblock {\em IEEE Transactions on Information Theory}, 55(7):3115--3147,
  2009.

\bibitem{tucci2011eigenvalue}
Gabriel~H. Tucci and Philip~A. Whiting.
\newblock Eigenvalue results for large scale random {V}andermonde matrices with
  unit complex entries.
\newblock {\em IEEE Transactions on Information Theory}, 57(6):3938--3954,
  2011.

\bibitem{tucci2014asymptotic}
Gabriel~H. Tucci and Philip~A. Whiting.
\newblock Asymptotic behavior of the maximum and minimum singular value of
  random {V}andermonde matrices.
\newblock {\em Journal of Theoretical Probability}, 27(3):826--862, 2014.

\bibitem{bazan2000conditioning}
Ferm{\'\i}n Baz{\'a}n.
\newblock Conditioning of rectangular {V}andermonde matrices with nodes in the
  unit disk.
\newblock {\em SIAM Journal on Matrix Analysis and Applications},
  21(2):679--693, 2000.

\bibitem{aubel2017vandermonde}
C{\'e}line Aubel and Helmut B{\"o}lcskei.
\newblock Vandermonde matrices with nodes in the unit disk and the large sieve.
\newblock {\em Applied and Computational Harmonic Analysis}, 2017.

\bibitem{ferreira1999super}
Paulo Ferreira, J.S.G.
\newblock Super-resolution, the recovery of missing samples and {V}andermonde
  matrices on the unit circle.
\newblock In {\em Proceedings of the Workshop on Sampling Theory and
  Applications, Loen, Norway}, 1999.

\bibitem{berman2007perfect}
Lihu Berman and Arie Feuer.
\newblock On perfect conditioning of {V}andermonde matrices on the unit circle.
\newblock {\em Electronic Journal of Linear Algebra}, 16(1):13, 2007.

\bibitem{gautschi1962inverses}
Walter Gautschi.
\newblock On inverses of {V}andermonde and confluent {V}andermonde matrices.
\newblock {\em Numerische Mathematik}, 4(1):117--123, 1962.

\bibitem{batenkov2020conditioning}
Dmitry Batenkov, Laurent Demanet, Gil Goldman, and Yosef Yomdin.
\newblock Conditioning of partial nonuniform fourier matrices with clustered
  nodes.
\newblock {\em SIAM Journal on Matrix Analysis and Applications},
  41(1):199--220, 2020.

\bibitem{kunis2020condition}
Stefan Kunis and Dominik Nagel.
\newblock On the condition number of vandermonde matrices with pairs of
  nearly-colliding nodes.
\newblock {\em Numerical Algorithms}, pages 1--24, 2020.

\bibitem{diederichs2019well}
Benedikt Diederichs.
\newblock Well-posedness of sparse frequency estimation.
\newblock {\em arXiv preprint arXiv:1905.08005}, 2019.

\bibitem{kunis2020smallest}
Stefan Kunis and Dominik Nagel.
\newblock On the smallest singular value of multivariate {Vandermonde} matrices
  with clustered nodes.
\newblock {\em Linear Algebra and its Applications}, 604:1--20, 2020.

\bibitem{batenkov2020spectral}
Dmitry Batenkov, Benedikt Diederichs, Gil Goldman, and Yosef Yomdin.
\newblock The spectral properties of vandermonde matrices with clustered nodes.
\newblock {\em Linear Algebra and its Applications}, 2020.

\bibitem{barnett2020exponentially}
Alex~H Barnett.
\newblock How exponentially ill-conditioned are contiguous submatrices of the
  fourier matrix?
\newblock {\em arXiv preprint arXiv:2004.09643}, 2020.

\bibitem{candes2008restricted}
Emmanuel~J. Candes.
\newblock The restricted isometry property and its implications for compressed
  sensing.
\newblock {\em Comptes Rendus Mathematique}, 346(9-10):589--592, 2008.

\bibitem{rudelson2008sparse}
Mark Rudelson and Roman Vershynin.
\newblock On sparse reconstruction from {F}ourier and {G}aussian measurements.
\newblock {\em Communications on Pure and Applied Mathematics},
  61(8):1025--1045, 2008.

\bibitem{liu2020resolution}
Ping Liu and Hai Zhang.
\newblock Resolution limit for line spectral estimation: Theory and algorithm.
\newblock {\em arXiv preprint arXiv:2003.02917}, 2020.

\bibitem{wedin1972perturbation}
Per-{\AA}ke Wedin.
\newblock Perturbation bounds in connection with singular value decomposition.
\newblock {\em BIT Numerical Mathematics}, 12(1):99--111, 1972.

\bibitem{meckes2007spectral}
Mark Meckes et~al.
\newblock On the spectral norm of a random toeplitz matrix.
\newblock {\em Electronic Communications in Probability}, 12:315--325, 2007.

\bibitem{adamczak2010few}
Rados{\l}aw Adamczak.
\newblock A few remarks on the operator norm of random toeplitz matrices.
\newblock {\em Journal of Theoretical Probability}, 23(1):85--108, 2010.

\bibitem{tropp2012user}
Joel~A. Tropp.
\newblock User-friendly tail bounds for sums of random matrices.
\newblock {\em Foundations of Computational Mathematics}, 12(4):389--434, 2012.

\bibitem{stoica1989music}
Petre Stoica and Arye Nehorai.
\newblock Music, maximum likelihood, and cramer-rao bound.
\newblock {\em IEEE Transactions on Acoustics, speech, and signal processing},
  37(5):720--741, 1989.

\bibitem{ottersten1991performance}
Bjorn Ottersten, Mats Viberg, and Thomas Kailath.
\newblock Performance analysis of the total least squares esprit algorithm.
\newblock {\em IEEE transactions on Signal Processing}, 39(5):1122--1135, 1991.

\bibitem{fannjiang2016compressive}
Albert~C. Fannjiang.
\newblock Compressive spectral estimation with single-snapshot {ESPRIT}:
  Stability and resolution.
\newblock {\em arXiv preprint arXiv:1607.01827}, 2016.

\bibitem{li2020super}
Weilin Li, Wenjing Liao, and Albert Fannjiang.
\newblock Super-resolution limit of the esprit algorithm.
\newblock {\em IEEE Transactions on Information Theory}, 66(7):4593--4608,
  2020.

\end{thebibliography}

\end{document}